\documentclass{LMCS}

\def\doi{8 (1:15) 2012}
\lmcsheading%
{\doi}
{1--27}
{}
{}
{Feb.~\phantom02, 2011}
{Mar.~\phantom02, 2012}
{}

\usepackage{hyperref}
\usepackage{amssymb,amsmath,xspace,enumerate,gensymb}

\usepackage{lineno}
\usepackage{graphicx,color, pstricks, pst-node, pst-plot}
\usepackage{xspace}
\usepackage{arydshln}
\usepackage{rotating}
\usepackage{listings}
\usepackage{soul}

\usepackage[comms,scclasses,proofs,amslib]{my}

 \usepackage[ruled]{algorithm2e}

\newcommand  {\problemdescr} [3] {
    \vspace{3mm}
    \def\Name{#1}
    \def\Input{#2}
    \def\Question{#3}
	  \hspace{5mm}\begin{tabular}{r p{11cm}r}%
	    \textit{Problem:} & \textsc{\Name} \\
	    \textit{Input:} & \Input \\
	    \textit{Question:} & \Question
	  \end{tabular}
    \vspace{3mm}
    }

\newcommand{\myd}{\ensuremath{\circ_d}}
\newcommand{\myc}{\ensuremath{\circ_c}}

\newcommand{\warrow}{\ensuremath{\leftarrow}}
\newcommand{\earrow}{\ensuremath{\rightarrow}}
\newcommand{\narrow}{\ensuremath{\uparrow}}
\newcommand{\sarrow}{\ensuremath{\downarrow}}

\newcommand{\Pro}{\ensuremath{\text{Pro}}\xspace}

\newcommand     {\DW}   {\myclass{FO}}

\newcommand     {\tpo}   {total preorder\xspace}
\newcommand     {\lo}   {linear order\xspace}
\newcommand     {\tpos}   {total preorders\xspace}
\newcommand     {\los}   {linear orders\xspace}

\newcommand     {\pup}   {\ensuremath{\uparrow}} 
\newcommand     {\pdown}   {\ensuremath{\downarrow}} 
\newcommand     {\pe}   {\ensuremath{\cdot}} 

\newcommand{\dw}{\ensuremath{\leftarrow}}
\newcommand{\de}{\ensuremath{\rightarrow}}
\newcommand{\dn}{\ensuremath{\uparrow}}
\newcommand{\ds}{\ensuremath{\downarrow}}
\newcommand{\dsw}{\ensuremath{\swarrow}}
\newcommand{\dse}{\ensuremath{\searrow}}
\newcommand{\dne}{\ensuremath{\nearrow}}
\newcommand{\dnw}{\ensuremath{\nwarrow}}

\newcommand{\precsucc}{S}

\renewcommand{\S}{\precsucc}
\newcommand{\N}{\precsucc}

\newcommand{\nS}{\neg \S}

\newcommand{\Ix}{\ensuremath{I_x}}
\newcommand{\Iy}{\ensuremath{I_y}}

\makeatletter 
\@ifundefined{proposition}{
 \newtheorem{proposition}[thm]{Proposition}
\newtheorem{theorem}[thm]{Theorem}

\newtheorem{corollary}[thm]{Corollary}

\newtheorem{myclaim}{Claim}
\newtheorem{example}[thm]{Example}

\newenvironment{proofof}[1]{\noindent{\bf Proof (of #1).}\enspace}{\qed}
}{}
\makeatother

\newcommand{\LPP}{\ensuremath{\text{2LPP}}}

\renewcommand {\calC}      {{\mathcal C}\xspace}
\renewcommand {\calD}      {{\mathcal D}\xspace}

\renewcommand {\calO}      {{\mathcal O}\xspace}
\renewcommand {\calP}      {{\mathcal P}\xspace}

\renewcommand {\calO}      {{\mathcal O}\xspace}

\renewcommand {\calS}      {{\mathcal S}\xspace}

\newcommand{\Prop}{\ensuremath{\text{\rm Prop}}}
\newcommand{\Seq}{\ensuremath{\text{\rm PSeq}}}

\newcommand{\PCP}{PCP\xspace}

\newcommand{\phit}[1]{\ensuremath{\varphi_{\text{#1}}}}

\setul{}{0.2mm}
\setstcolor{red}

\rnc{\todo}[1]{\ \\ {\color{red} \fbox{\parbox{\linewidth}{{\sc ToDo}:\\  #1}}}}

\newcounter{CommentCounter}

\newcommand{\FinOrd}{\ensuremath{\text{FinOrd}}}

\begin{document}

%
\title{Two-Variable Logic with Two Order Relations }

\author[T.~Schwentick]{Thomas Schwentick\rsuper a}	
\address{{\lsuper{a,b}}TU Dortmund University}	
\email{\{thomas.schwentick, thomas.zeume\}@cs.uni-dortmund.de}  

\author[T.~Zeume]{Thomas Zeume\rsuper b}
\address{\vskip-6 pt}	

\thanks{{\lsuper{a,b}}We acknowledge the financial support of the Future and Emerging Technologies (FET) programme
  within the Seventh Framework Programme for Research
  of the European Commission, under the FET-Open
  grant agreement FOX, number FP7-ICT-233599. We further acknowledge
  the financial support by the German DFG under grant SCHW 678/4-1.}	
%
%
%



\begin{abstract}
    It is shown that the finite satisfiability problem for
    two-variable logic  over structures with
    one \tpo relation, its induced successor relation, one \lo
    relation and some further unary relations is  \EXPSPACE-complete. Actually, \EXPSPACE-completeness already holds for structures that do not include the induced successor relation. As a
   special case, the \EXPSPACE upper bound applies to two-variable logic over
   structures with two    linear orders. A further consequence is  that satisfiability of
    two-variable logic over data words with a linear order on
    positions and a linear order and successor relation on the data is
    decidable in \EXPSPACE.

    As a complementing result, it is shown that over structures
    with two \tpo relations as well as over structures with one \tpo
    and two \lo relations, the finite  satisfiability problem for
    two-variable logic is undecidable. 
\end{abstract}

\maketitle

\section{Introduction}  

First-order logic restricted to two-variables ({\em two-variable logic} or
$\FO^2$ in the following) is generally known to be reasonably expressive for
many purposes. In contrast to full first-order logic, its satisfiability and its finite
satisfiability problem are decidable \cite{MortimerOn75}, in fact they
are \NEXPTIME-complete \cite{GradelKV97}. 

However, if one is interested in (finite or general)  satisfiability
of $\FO^2$ over structures with a particular property $P$ these general
results can only be applied if $P$ is expressible in
$\FO^2$. Unfortunately, there are some simple properties like transitivity of a binary relation that
cannot be expressed in $\FO^2$. In particular, in $\FO^2$ it can
neither be expressed that a given binary relation is a linear order
nor that it is an equivalence relation. Thus, the results from
\cite{MortimerOn75, GradelKV97} do not help for satisfiability
of $\FO^2$ over (finite or general) structures with linear orders or equivalence relations, 

In \cite{OttoTwo01} it was shown that it can be decided in \NEXPTIME whether a given
$\FO^2$ sentence has a model (or whether it has a finite model) in
which a particular relation symbol is interpreted by a linear
order. On the other hand, in the presence of eight binary symbols that
have to be interpreted as linear orders it is undecidable. Note that
in these results the formulas might use further relation symbols of
arbitrary arity for which the possible interpretations are
unrestricted. 

The problem of deciding finite satisfiability\footnote{As this article only
  deals with finite structures we henceforth only mention results on finite satisfiability. } 
was shown to be \NEXPTIME-complete over structures with one equivalence
relation and undecidable over structures with three equivalence
relations in \cite{KieronskiOttoSmall05}.  In \cite{KieronskiTenderaOn09} it was
shown that over structures with two equivalence relations the problem is
decidable in triply exponential nondeterministic time.

In this article we study two-variable logic over structures with linear
orders and total preorders. A total preorder $\precsim$ is basically an
equivalence relation $\sim$ whose equivalence classes are ordered by
$\prec$. Total preorders can therefore encode equivalence relations as
well as linear orders and in this sense they generalize both types of
relations. It should be stressed that in our results, structures may have
an arbitrary number of additional unary relations but {\em
  no} further non-unary relations.

Our motivation stems from the context of so-called data words. A
{\em data word} is a word, that is, a finite sequence of symbols
from a finite alphabet, but besides a symbol, every position also
carries a value from a possibly infinite domain. An interest in data
words and data trees arises from applications in database theory,
where XML documents can be modeled by data trees in which the symbols
correspond to the tags and the data values to text or attribute
values. On the other hand, (infinite) data words can also be
considered as traces of  computations in a distributed environment,
where symbols correspond to states of processes and data values
encode process numbers. Recently many logics and automata models have
been considered for data words and data trees (see
\cite{SegoufinAutomata06} for a gentle introduction). 
 
First-order logic on data words (with a
linear order on positions and equality on data values)  is undecidable, even for
formulas with three variables
\cite{BojanczykMuscholl+Two-Variable06}. 

On the other hand, finite satisfiability of two-variable logic on data words is
decidable. More precisely, it was shown in
\cite{BojanczykMuscholl+Two-Variable06} that it is decidable even in the
setting where data
words are equipped with a linear order and its corresponding successor
relation on positions and equality on data values. However, the
complexity is unknown but basically equivalent to the open complexity of Petri
  net reachability. It was further shown in
  \cite{BojanczykMuscholl+Two-Variable06} that the problem is
  $\NEXPTIME$-complete without the successor relation on the positions
  and that it becomes undecidable if the data values are equipped with
  a linear order (in the presence of linear order and successor
  relation on positions).

Data words are closely related with finite structures with order relations.
More precisely, data words can be represented by
\begin{iteMize}{$\bullet$}
\item some unary relations to encode symbols at positions,
\item a linear order (on the positions), possibly its induced successor
relation, and
\item an equivalence relation corresponding to data equality between positions.
\end{iteMize}
In this context, an additional linear order
on data values can be represented by a total preorder $\precsim$. All
positions with a particular value constitute an equivalence class and
these classes are ordered by the linear order on the values.  
It is exactly this setting which triggered our study of structures with a
linear order, a total preorder, and
some unary relations.

\subsubsection*{Results.} We show that finite satisfiability of
two-variable logic over structures with a linear order, a total
preorder,  and unary relations is
$\EXPSPACE$-complete. Consequently, finite 
satisfiability of $\FO^2$ over data words with a linear order (but no
successor relation)  on the
positions and a linear order on the data
values can be decided in exponential space as well.
As it can be expressed in $\FO^2$ that a total preorder is a linear order,
the corresponding problem with two linear orders (and no total
preorder) is also solvable in \EXPSPACE.  Thus, the gap between one
and eight linear orders that was left in the work of
Otto \cite{OttoTwo01} is narrowed.

The upper bound even holds when the preorder is accompanied by a
partial successor relation (the precise definition of this
  notion will be given in Section \ref{sec:preliminaries}
  ). As a corollary, satisfiability of $\FO^2$ on data words is in
  $\EXPSPACE$ even if the linear order on the data values is accompanied with a corresponding partial successor relation.


The upper bounds are obtained by a reduction to finite satisfiability of {\em semi-positive}
$\FO^2$ sentences over sets of labeled points in the plane, where points can be compared by
their relative position with respect to the directions
$\nwarrow,\nearrow,\swarrow,\searrow,\warrow,\earrow$. Furthermore two
points can be tested for being on successive horizontal lines.
Semi-positive formulas are in negation normal form and do not
contain negated ``direction atoms''. For a precise definition we
refer to Section \ref{sec:main}. Finite satisfiability of {\em semi-positive}
$\FO^2$ over such point sets can in turn be reduced in exponential
time to a constraint problem
for labeled points in the plane with \PSPACE-complexity. 
The $\EXPSPACE$ lower bound is by a reduction
from exponential width corridor tiling. 

Finally, we show by reductions from the Post Correspondence Problem (\PCP) problem that finite satisfiability of $\FO^2$ over
structures\footnote{Additional unary relations are again allowed.} with two total preorders and over
structures with two linear orders and a total preorder is
undecidable. 
%
\subsubsection*{Organization.} After some basic definitions in Section
\ref{sec:preliminaries}, we prove the \EXPSPACE upper bound in Section
\ref{sec:main} and 
all lower bounds in Section \ref{sec:lower}.  We conclude with Section
\ref{sec:conclusion} where we discuss research directions and related
work on compass and interval logics.

\subsubsection*{Related work.}

As mentioned before, also other logics for data words besides $\FO^2$
have been studied. As an example we mention the ``freeze''-extension
of $\text{LTL}$ studied, e.g. in
\cite{DemriLazicLTL09} and 
\cite{FigueiraSegoufinFuture-Looking09}. The latter paper is more
closely related to our work as it considers a restriction of $\text{LTL}$ without
the $X$-operator.  Amaldev Manuel has recently proved
decidability and undecidability results for $\FO^2$-logic over
structures with orders \cite{Manuel10}. However, in his work
structures have at least two successor relations but no linear orders, hence neither
results nor techniques translate from his work to ours nor vice
versa. 

\subsubsection*{Acknowledgements.} We thank Jan van den Bussche for
stimulating discussions and Daniela Huvermann for careful proof reading.
\section{Preliminaries}\label{sec:preliminaries}

In this article, we only consider \textit{finite} structures.
We are interested in three kinds of finite structures:
ordered structures, sets of labeled points in the plane and data
words.

In the following, $\nat$ denotes the set of natural numbers and $\rat$
the set of rationals. 

\subsubsection*{Ordered Structures.} We first fix our notation concerning
order relations. A {\em \tpo} $\precsim$ is a transitive, total
relation, that is, $u \precsim v$ and $v \precsim w$ implies $u
\precsim w$ and for every two elements $u,v$ of a structure $u
\precsim v $ or  $v \precsim u$ holds. In particular, every total preorder is reflexive, that is $u \precsim u$ holds for every $u$. A {\em \lo} $\leq$ is a
antisymmetric \tpo, that is, if $u \leq v$ and $v \leq u$ then $u = v$.
Thus, the essential difference between a \tpo and a \lo is that the
former allows that for two distinct elements $u$ and $v$ both  $u
\precsim v$ and $v \precsim u$ hold.  We call two such elements {\em
  equivalent with respect to $\precsim$} and write $u \sim
v$. We write $u\prec v$ if $u\precsim v$ and $u\not\sim v$.

A \tpo can be viewed as an equivalence relation $\sim$
whose equivalence classes are linearly ordered by $\prec$. Clearly, every \lo is
a \tpo with equivalence classes of size one. 

We define the  \textit{induced successor relation} $\S_\precsim$ of a total preorder $\precsim$ 
 as follows. For
  two elements $u,v$,  $\S_\precsim(u,v)$ if $u\prec v$,
  and there is no element $w$ such that $u\prec w \prec v$. 
A \emph{partial successor relation} $S$ is a sub-relation of $\S_\precsim$ such that $u\sim u'$, $v\sim v'$ and $S(u,v)$ imply $S(u',v')$. Thus, a partial
  successor relation is a sub-relation of $\S_\precsim$ that is derived
from a sub-relation of the successor relation on the equivalence
classes of $\precsim$. We say that $\S_\precsim$ is \textit{complete} to distinguish it from
(truly) partial successor relations.

We use binary relation symbols $<, <_1, <_2, \ldots$ that are always
interpreted as \los, binary relation symbols
$\precsim,\precsim_1,\precsim_2,\ldots$ that are interpreted as \tpos, and
binary relation symbols $S, \S_1, \S_2, \ldots$
that are interpreted as (partial or complete) successor relations. 
 We note that $\sim$ and $\prec$ can be expressed
in two-variable logic, given $\preceq$. 

In this article, an \emph{ordered structure} is a finite structure with
non-empty universe and 
some \los, some \tpos, some successor relations and some unary relations. We
always allow an unlimited number of unary relations and specify the
numbers of allowed \los and \tpos explicitly. 

We denote classes of ordered structures by the notation
  $\FinOrd(\calO)$ where $\calO$ indicates the orders and successor
  relations following the above conventions. Here, corresponding
  relations are grouped together by square brackets.
For example, by
$\FinOrd(\leq_1,[\precsim_2, S_2])$ we denote
the set of finite structures with one \lo and one
\tpo together with a corresponding partial successor relation.

\subsubsection*{Sets of Labeled Points.} As mentioned before, we also consider finite sets of labeled points. Let $\Prop=\{e_1,\ldots,e_k\}$ be a set of propositions.  A
{\em $\Prop$-labeled point} $p$ is a point in $\nat^2$ in which 
propositions $e_1,\ldots,e_k$ may or may not 
hold. We refer to the $x$-coordinate and the $y$-coordinate of a point
$p$ by $p.x$ and $p.y$, respectively.
We simply say {\em point} if $\Prop$ is understood from the
context. We do not allow different points $p\not=q$ at the same
position, that is, if $p.x=q.x$ and $p.y=q.y$ then $p = q$.
We say that a set $\Prop$ of labeled points is
\textit{contiguous}, if the $y$-coordinates of $\Prop$ constitute an
interval\footnote{It is not required that every number
  occurs only once as the $y$-coordinate of a point.} in $\nat$. 

For a finite set $P$ of labeled points, we write $p \dnw q$ if
$p.x>q.x$ and $p.y < q.y$, that is if $q$ is in the northwest of $p$.
Likewise for $\dw, \dsw, \ds, \dse, \de, \dne, \dn$.  We write
$\Iy(p,q)$ if $p.y + 1 = q.y$ and $\Ix(p,q)$ if $p.x + 1 = q.x$. Let $\calD = \{\dne, \dn, \dnw, \dw, \dsw, \ds, \dse,
\de\}$ denote the set of {\em directions}. We denote $\calD$ without
$\ds$ and $\dn$ by $\calD_-$. 


\subsubsection*{Point Sets versus Ordered Structures.} There is a strong
connection between sets of labeled points and ordered structures with
two total preorders.

With every finite ordered structure $\calO$ with universe $U$ and
relations $\precsim_1$, $\precsim_2$, their induced successor
relations $\S_1,\S_2$ and  some unary relations one can
associate a finite point set $M$ in the plane and a bijection
$\pi$ such that the following statements hold.
\begin{iteMize}{$\bullet$}
\item $u\prec_1 v$ in $\calO$ if and only if $\pi(u)\dne \pi(v)$ or $\pi(u)\de
\pi(v)$ or $\pi(u)\dse \pi(v)$ in $M$.
\item $u\prec_2 v$ in $\calO$ if and only if $\pi(u)\dnw \pi(v)$ or
  $\pi(u)\dn \pi(v)$ or $\pi(u)\dne
\pi(v)$ in $M$.
\item $u\sim_1 v$ in $\calO$  if and only if $\pi(u)\dn \pi(v)$ or $\pi(u)\ds
\pi(v)$ or $\pi(u) = \pi(v)$ in $M$.
\item $u\sim_2 v$ in $\calO$  if and only if $\pi(u)\dw \pi(v)$ or $\pi(u)\de
\pi(v)$ or $\pi(u) = \pi(v)$ in $M$.
\item $\S_1(u,v)$ in $\calO$ if and only if $\Ix(\pi(u), \pi(v))$ in $M$.
\item $\S_2(u,v)$ in $\calO$ if and only if $\Iy(\pi(u), \pi(v))$ in $M$.
\end{iteMize}

To this end, we can assign to every equivalence class of $\precsim_1$
a natural number in increasing order, likewise for $\precsim_2$. Then,
$\pi(u).x$ is just the number of $u$'s
$\sim_1$-class and $\pi(u).y$ is the number of its
$\sim_2$-class. Note that this construction might yield a multiset as
there could be elements $u\not=v$ such that $u\sim_1 v$ and $u\sim_2
v$. However, in the following this case will not occur as we will
require that $\precsim_1$ is a linear order.

Similarly, from every labeled point set $M$ in $\nat^2$, an
ordered structure $\calO$ (with universe $M$) can be obtained by defining
\begin{iteMize}{$\bullet$}
\item $u\precsim_1 v$ if $u.x\le v.x$ and
\item $u\precsim_2 v$ if $u.y\le v.y$
\end{iteMize}
and by defining $S_1$ and $S_2$ accordingly.

Thus, points are equivalent with respect to
  $\precsim_2$ if they have the same $y$-coordinate.  

 

\subsubsection*{Data Words.} We are interested in a particular kind of
ordered structures, \textit{data words with ordered data}. In a
nutshell, a data word with ordered data is a
string in which every position carries a label from a finite alphabet {\em
  and} a value from a potentially infinite, ordered domain. In this paper, this
domain will always be $\nat$. As only a finite number of values can occur in a finite
  data strings this is not a restriction. More formally, a \textit{data word} $s$
over alphabet $\Sigma$ is a  finite sequence $(\sigma_1,d_1),\ldots,(\sigma_n,d_n)$ 
where  $\sigma_i\in\Sigma$ and  $d_i\in\nat$, for every $i$.  Such a data word can be
represented in a natural way by an ordered structure whose universe is
the set $\{1,\ldots,n\}$ of positions of $s$ and which is equipped
with a linear order $\le_1$ and a successor relation $S_1$ on the positions. To represent
the linear order on the data values the structure may have a total
preorder $\precsim_2$. Thus, $i\precsim_2 j$ if $d_i\le d_j$. Note
that if the same data value occurs at different positions, 
$\precsim_2$ may indeed have non-singleton equivalence
classes. Furthermore, the structure has one unary relation for each
symbol of $\Sigma$. In logical formulas we simply write $\sigma(x)$ to
denote that position $x$ carries symbol $\sigma$. Note that data words
are a special kind of ordered structures as every position is
contained in exactly one of the unary relations. 

In this article, we do not consider data words with the
successor relation on positions. However, we allow a partial
successor relation $S_2$ of the total preorder which translates to the
successor relation on the data values as $(i,j)\in S_2$ if and only if
$d_i+1=d_j$.

From the point of view of finite structures there is only one
difference between ordered structures and data words: in data words
every position carries exactly one symbol from some finite alphabet
$\Sigma$, whereas ordered structures in general allow several unary
relations that need not be disjoint. 
  We thus add $\Sigma$ to our notation of ordered structures to
  indicate that there is one unary relation per symbol in $\Sigma$ and
  the unary relations partition the universe. For example, write 
$\FinOrd(\Sigma,\leq_1, [\precsim_2,S_2])$  to denote predicate logic over
data words with ordered data, alphabet $\Sigma$ and domain $\nat$
without successor on positions but with successor on values.

The correspondence between ordered structures and sets of labeled
points naturally translates to a correspondence between data words
with ordered data and sets of labeled points. Figure
\ref{fig:dwlabpoint} shows the point set corresponding to the data
word $w = \biggr (\begin{array}{lllllll}
      a & a & a & b & b\\
      1 & 3 & 4 & 6 & 3 \\
\end{array}\biggl )$.

  \begin{figure}[t]
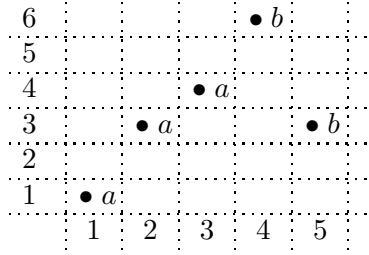
 \centering
  \setlength{\dashlinegap}{3pt}
  \setlength{\dashlinedash}{0.5pt}
  \newdimen\mywidth
  \mywidth=0.4cm
      \begin{tabular}{p{\mywidth}: p{\mywidth}: p{\mywidth}: p{\mywidth}: p{\mywidth}: p{\mywidth} : p{\mywidth}}
	 6 & &  & &  \centering{$\bullet \; b$}& &\\
	\hdashline
	 5 & &  & & & & \\
	\hdashline
 	 4 & & & \centering{$\bullet \; a$}& & &\\
	\hdashline
	 3 & & \centering{$\bullet \; a$}& & & \centering{$\bullet \; b$} &\\
	\hdashline
	 2 &  & & &   & &\\
	\hdashline
	 1 & \centering{$\bullet \; a$}& &  & & &\\
\hdashline
& \centering{1} & \centering{2} & \centering{3} & \centering{4} & \centering{5}\\ 
      \end{tabular}
    \caption{Finite set of labeled points corresponding to the data word $(a,1)(a,3)(a,4)(b,6)(b,3)$. Columns are ordered by $<_1$ and rows are ordered by $\precsim_2$.
\label{fig:dwlabpoint}}
  \end{figure}



\subsubsection*{Logic.} 
\textit{Two-variable logic} is the
restriction of predicate logic to formulas that only use (at most)
two variables $x$ and $y$. We denote two-variable logic by $\FO^2$. 

We denote two-variable logic on a restricted set $\calS$ of structures
by $\FO^2(\calS)$. As an abbreviation, we denote
$\FO^2(\FinOrd(\calO))$, for a set $\calO$ of predicates by $\FO^2(\calO)$.

If $\calD$ is a set of directions, we denote by $\FO^2(\calD)$
two-variable logic with binary relations from
$\calD$ and unary relations from $\Prop$.
We interpret $\FO^2(\calD)$ logic over non-empty sets of
$\Prop$-points in $\nat^2$ where $\Prop$ has a proposition $e_i$, for
every unary relation symbol $U_i$.

\section{Upper Bounds for Two-variable Logic on Ordered Structures}\label{sec:main}

The following theorem states our main result.

\begin{theorem} \label{theo:twodcompl}
  Finite satisfiability of $\FO^2(\leq_1, [\precsim_2, S_2])$ is \EXPSPACE-complete.
\end{theorem}

\begin{corollary} \label{cor:dw}
Finite satisfiability of $\FO^2$ over data words with a linear order
  on the positions and a linear order and a corresponding successor
  relation on the data values can be decided in \EXPSPACE.

\end{corollary}

It is worthwhile to compare Corollary \ref{cor:dw} with the following
result from \cite{BojanczykMuscholl+Two-Variable06}.
\begin{proposition}\label{prop:succundec}
  Finite satisfiability of $\FO^2$ over data words with a linear order
  and its  corresponding successor
  relation on the positions and a linear order on the data values is undecidable.
\end{proposition}
The two results might appear contradictory at first sight. However,
translated into the language of orders the two settings are indeed
different. In Corollary \ref{cor:dw} there is a 
total preorder with a corresponding successor relation plus a linear order, whereas in
Proposition \ref{prop:succundec} there is a linear order with its
underlying successor relation plus a total preorder. We can conclude
that the additional total preorder yields more expressive power in
that it allows to encode Post's correspondence problem in a
two-variable fashion as shown in \cite{BojanczykMuscholl+Two-Variable06}..

The lower bound of Theorem \ref{theo:twodcompl} will be shown in the
following section as Theorem \ref{theo:expspacetwoorders}.
The upper bound proof is given in this
section and consists of three main steps.
\begin{enumerate}[({3}.1)]
\item We first reduce finite satisfiability of $\FO^2(\leq_1,
  [\precsim_2, S_2])$ to finite satisfiability of
  $\FO^2(\calD)$-formulas of a syntactically restricted form in polynomial time.
\item The latter problem is then reduced to the two-dimensional
labeled point problem ($\LPP$), which will be defined below, in exponential time.
\item Finally, we show that
$\LPP$ can be solved in polynomial space. 
\end{enumerate}

\subsection{From ordered structures to labeled point sets}

By $\FO^2(\calD_-,I_y)$ we denote two-variable
logic with binary atoms $x\myd y$ using directions $\myd$
from $\calD_-$ (in infix notation) and binary predicate $I_y$. A
sentence $\varphi$ of this logic is {\em
  semi-positive} if it fulfills the following conditions.
\begin{enumerate}[(i)]
\item $\varphi$ is in negation normal form
(NNF).
\item There are no negated literals  $\neg(x\myd y)$ in $\varphi$.
However, the occurrence of negated atoms $x\not= y$ is not restricted.
\item Literals $I_y(x,y)$ and $\neg I_y(x,y)$ only appear in
  conjunction with positive atoms $\myd \in \calD_-$.  
\end{enumerate}

The correspondence between ordered structures and finite sets of
labeled points can be exploited to show the following result.

\begin{proposition}\label{prop:topoints}
  For each $\FO^2(\leq_1 [\precsim_2, S_2])$-sentence $\varphi$ a semi-positive
  $\FO^2(\calD_-, I_y)$-sentence $\psi$ can be computed in polynomial
  time such that $\psi$ and $\varphi$ are equivalent with respect to
  finite satisfiability over structures with a linear order and a
  total preorder with induced successor relation.
\end{proposition}

\begin{myproof}
We assume without loss of generality that $\varphi$ is given in
negation normal form. The target formula $\psi$ is of the form
$\chi\land\varphi'$ where $\chi$ ensures that no
two points are on the same vertical line and $\varphi'$ is the actual
translation of $\varphi$ into the setting of labeled points.

More precisely, 
\[
\chi=\forall x \forall y\;\; x=y\lor\bigvee_{\myd\in\calD_-} x\myd y,
\]
and $\varphi'$ is obtained from $\varphi$ by translating its binary literals as
follows.
\[\small
  \begin{array}[c]{|c|c|}
    \hline \varphi & \varphi'\\\hline
x\le_1 y & x=y \lor x\dne y \lor x \de y \lor x \dse y\\
\neg(x\le_1 y) & x\dnw y \lor x \dw y \lor x \dsw y\\
x\precsim_2 y & x=y \lor x\de y \lor x \dw y \lor x
  \dne y \lor x \dnw y\\
\neg(x\precsim_2 y) & x\dse y \lor x \dsw y\\
S_2(x, y) & (I_y(x,y) \wedge x \dne y) \lor (I_y(x,y) \wedge x \dnw y) \\
\neg S_2(x,y) &  (\neg I_y(x,y) \wedge x \dne y) \lor (\neg I_y(x,y)
\wedge x \dnw y) \lor x=y \lor x \dw y \lor x \de y \lor x \dsw y \lor x \dse y \\

\hline
  \end{array}
\]

The translation of literals as $y\le_1 x$ is analogous.

The correctness of this translation relies on the fact that $\chi$
ensures that no two points are on the same vertical line. 
For example, for $x\le_1 y$, no disjuncts $x\narrow y$ or $x\sarrow y$ are
needed. It is easy to verify that $\psi$ is indeed semi-positive. 
\end{myproof}

\subsection{From Two-Variable Logic to Two-Dimensional Constraints}

We next define the 2-dimensional labeled point problem, $\LPP$. 
For an alphabet $\Sigma$, a {\em $\Sigma$-labeled
  point} $p$ is an element from $\nat^2\times\Sigma$. 
Of course, $\Prop$-labeled points can be considered as
  $\Sigma$-labeled points with $\Sigma=2^\Prop$. However, the number of
  symbols is exponential in the number of propositions. Likewise,
  $\Sigma$-labeled points can be encoded as points over a 
   suitable set of propositions but in that direction the encoding is
   less canonical. One could, e.g., use one proposition per symbol of
   $\Sigma$ or a logarithmic number of propositions.
 We write $p.l$ for the label of a $\Sigma$-labeled point $p$.

A \textit{directional constraint} is a pair $(\myd,s_d)$  from $\calD\times
\{\S, \nS\}$. 
We denote the set of directional constraints by $\calC$. By $\calC_-$
we denote the set of directional constraints $(\myd,s_d)$ with $\myd\in\calD_-$. 

We say a pair $(p,q)$ of $\Sigma$-labeled points \textit{satisfies} a directional constraint $d = (\myd, s_d)$, if $p \myd q$ and the following conditions are fulfilled. 
\begin{iteMize}{$\bullet$}
  \item $p.y+1 = q.y$, if $s_d = S$.
  \item $p.y +1 \not= q.y$, if $s_d = \nS$.
\end{iteMize}

Note that the constraints $(\de, \S)$ and $(\dw, \S)$ are not satisfied for any pair of points. By $\not=$  we denote an \textit{inequality constraint} (and there is
only one such constraint). A pair $(p,q)$ satisfies the inequality
constraint $\not=$ if $p\not=q$. 

A \textit{position
constraint} is either a directional or an inequality constraint.


An {\em existential constraint ($\exists$-constraint)} is a pair
$(\sigma,E)$ where $\sigma \in \Sigma$ and $E$ is a possibly empty set
of pairs $(\tau,d)$, where $\tau \in \Sigma$ and $d$ is a position constraint from $\calC_-\cup\{\not=\}$. For a set $\calP$  of $\Sigma$-labeled points and a point $p$, we say $p$
{\em satisfies} an $\exists$-constraint  $(\sigma,E)$ if either
$p.l\not=\sigma$ or there
is $q\in \calP$ such that, for some $(\tau,d)$ in $E$, $q.l=\tau$ and $(p,q)$ \mbox{satisfies 
$d$.} 

A {\em universal constraint ($\forall$-constraint)} is a tuple
$(\sigma,\tau,d)$ where $\sigma,\tau\in\Sigma$ and $d$ is a
directional constraint. A pair $(p,q)$ of points from $M$ {\em satisfies} a
$\forall$-constraint  $(\sigma,\tau,d)$ if it is {\em not the case}
that $p.l=\sigma$, $q.l=\tau$, and $(p,q)$ satisfies $d$.

An input $L=(\Sigma,C_\exists,C_\forall)$  to the {\em two-dimensional
  labeled point problem ($\LPP$)} consists of an alphabet
$\Sigma$, a set $C_\exists$ of existential constraints and a set $C_\forall$
of universal constraints. 
A non-empty set $M\subs \nat^2\times\Sigma$ is a {\em solution} of
$L$ if every point of $M$ satisfies all constraints from $C_\exists$
and every pair of distinct points satisfies all constraints from
$C_\forall$. It should be noted that $C_\exists$ specifies required
patterns whereas $C_\forall$ specifies forbidden patterns.

\begin{proposition}\label{prop:tolpp}
 From every
  semi-positive $\FO^2(\calD_-,I_y)$-sentence $\varphi$ an instance $L$ of
  $\LPP$ can be computed in exponential time such that
  $\varphi$ is finitely satisfiable if and only if $L$ has
  a finite solution.
\end{proposition}
\begin{myproof}
First, $\varphi$ can be translated into a semi-positive
$\FO^2(\calD_-,I_y)$-sentence  $\varphi'$ in Scott 
normal form (SNF)
\[
\forall x \forall y\; \psi(x,y) \wedge \bigwedge_{i=1}^m \forall x \exists y\; \psi_i(x,y),
\]
that has a finite model if and only if $\varphi$ has a finite
model. The translation can be done in a way that ensures that
\begin{enumerate}[(1)]
\item $\psi$ and all $\psi_i$ are quantifier-free semi-positive formulas,
\item $\varphi'$ is of linear size in the size of $\varphi$,
\end{enumerate}
In general, $\varphi'$ uses more unary relation symbols than
$\varphi$.

The translation mimics the proof of Theorem 2.1 in
\cite{GraedelO1999}. However, we need to be a bit careful to get
$\varphi'$ semi-positive.

The transformation is done in several rounds. After each round, a
formula of the form $\theta_i\land\varphi_i$ is obtained that is
equivalent to $\varphi$ with respect to (finite) satisfiability. Here,
$\theta_i$ is already in SNF and $\varphi_i$  has $i$ quantifiers less
than $\varphi$. Furthermore, $\varphi_i$ uses additional new unary
predicates $P_1,\ldots,P_i$ and is semi-positive. 
Initially, we have $\theta_0=\top$ and
$\varphi_0=\varphi$. In round $i$, a subformula $\chi_i(x)$ of the form $\exists
y\; \rho_i(x,y)$ or the form $\forall y\; \rho_i(x,y)$ with quantifier-free formula $\rho_i$ is
chosen from $\varphi_{i-1}$. 
 In the former case,
 \begin{iteMize}{$\bullet$}
 \item $\theta_i=\theta_{i-1} \land \forall x \exists y (P_i(x) \ra
   \rho_i(x,y))$ and
 \item $\varphi_i$ is obtained from $\varphi_{i-1}$ by replacing $\exists
y\; \rho_i(x,y)$ with $P_i(x)$.
 \end{iteMize}

In the latter case,
\begin{iteMize}{$\bullet$}
\item $\theta_i=\theta_{i-1} \land \forall x \forall y (P_i(x) \ra
   \rho_i(x,y))$ and
 \item $\varphi_i$ is obtained from $\varphi_{i-1}$ by replacing $\forall
y\; \rho_i(x,y)$ with $P_i(x)$.
\end{iteMize}

Clearly, this process has as many rounds as the initial formula
$\varphi$ has quantifiers. It is easy to see that each $\theta_i$ is
semi-positive and in SNF and each $\varphi_i$ is semi-positive as well. Let
$\theta_k\land\varphi_k$ be the resulting formula.  $\varphi_k$
does not contain any quantifiers and is equivalent to $\forall x
\forall y\; \varphi_k$. Therefore, $\varphi'=\theta_k\land \forall x
\forall y\; \varphi_k$ is in SNF (after combining all
$\forall\forall$-conjuncts into one
$\forall\forall$-formula). Furthermore, the size of $\varphi'$ is
linear in the size of $\varphi$. However, the signature has been
extended by $k$ unary relation symbols.

It remains to show that $\varphi'$ is finitely satisfiable if and only
if $\varphi$ is finitely satisfiable. To this end,
we show by induction on $i$ that, for every $i>0$, $\theta_i\land\varphi_i$ is finitely
satisfiable if and only if $\theta_{i-1}\land\varphi_{i-1}$ is finitely
satisfiable.
Indeed, if $A_{i-1}\models \theta_{i-1}\land\varphi_{i-1}$, a
model\footnote{Here, $(B,P)$ denotes the extension of a structure $B$
  by a relation $P$.} $A_i=(A_{i-1},P_i)$
for   $\varphi_i$ can be
obtained from $A_{i-1}$ by letting $P_i=\{a\mid A \models \chi_i(a)\}$. On the
other hand, assume that $A_i=(A_{i-1},P_i)$ is a model for
$\theta_i\land\varphi_i$ in the case that $\chi_i(x)$
is of the form $\exists y\; \rho_i(x,y)$. Thus, $A_i\models\theta_i$ guarantees
that for every $a$ with $A_i\models P_i(a)$ the formula $\exists
y\; \rho_i(a,y)$ is satisfied by $A_i$ and thus also $A_{i-1}\models \exists
y\; \rho_i(a,y)$ (as $P_i$ does not occur in $\rho_i$). 
As $\varphi_{i-1}$ is in NNF, the occurrence of $\exists
y\; \rho_i(a,y)$ is not in the scope of any negations, therefore 
$A_i\models\varphi_i$ implies $A_{i-1}\models \varphi_{i-1}$. 
 Furthermore, $A_i\models\theta_i$ also
implies  $A_{i-1}\models\theta_{i-1}$. If $\chi_i(x)$
is of the form $\forall y\; \rho_i(x,y)$ the reasoning is analogous. This completes the proof that $\varphi'$ is finitely satisfiable if and only
if $\varphi$ is finitely satisfiable.

In order to continue the translation into a \LPP-instance, we next
describe how to obtain $\forall$-constraints from $\psi$. A \textit{full
atomic type for $x$ and $y$} consists of
\begin{iteMize}{$\bullet$}
\item a full atomic type $\sigma$
for $x$,
\item a full  atomic type $\tau$
for $y$, and
\item a conjunction of binary literals concerning $x$ and $y$ which
  can be either
  \begin{iteMize}{$-$}
  \item $x=y$ or
  \item 
the conjunction of a direction atom $x\myd y$ with
    $\myd\in\calD$ and one of the formulas
    $S(x,y)$ or $S(y,x)$
or $\neg S(x,y)\land \neg S(y,x)$.  
  \end{iteMize}
\end{iteMize}
It should
be noted that we do not allow $x\not=y$ in full atomic types as
this literal is equivalent to the disjunction of all other possible
relationships besides $x=y$.  

We further note that each full atomic type which does not contain the
  formula $x=y$ corresponds to a universal constraint. For instance,
  $\sigma(x)\land\tau(y)\land x\dne y\land \neg S(x,y)\land \neg
  S(y,x)$ corresponds to $(\sigma,\tau,(\dne,\neg 
S))$ and $\sigma(x)\land\tau(y)\land x\dse
y\land S(y,x)$ corresponds to $(\tau,\sigma, (\dnw,S))$.  The reader
should note the direction
change and the switch of $\sigma$ and $\tau$ in the latter example. It
is due to the translation of the atom $S(y,x)$.

For every full atomic binary or unary type, $\psi$ determines whether this type is
allowed or forbidden in a model of $\varphi'$. We let $\Sigma$ be the set of all
\textit{unary} atomic types $\sigma$ of $x$ for which $\sigma(x)\land\sigma(y)\land
x=y$ is \textit{allowed}. Furthermore, we let $C_\forall$ be the set of all
$\forall$-constraints that are obtained from the \textit{forbidden binary types}. 
Clearly, the size of $C_\forall$ is at most exponential in $\varphi$.

It should finally be noted that, although $\varphi$ only uses directions
  from $\calD_-$, the resulting universal constraints might contain
  the directions $\dn$ or $\ds$. This is for instance necessary in the case that $\psi$ implies
$\forall x \forall y\;(\sigma(x)\land\sigma(y) \ra x = y)$, since here universal constraints $(\sigma, \sigma, \myd)$ are needed for all $\myd \in \calD$.

By transforming every $\psi_i$ into DNF and some additional simple steps, every conjunct $\forall x \exists y\;
\psi_i(x,y)$ of $\varphi'$ can be  rewritten as
\[
\forall x\; \bigwedge_{j=1}^K\big(\sigma_j(x) \rightarrow
 \exists y\;\bigvee_{\ell=1}^M (\tau_{j\ell} \land\psi_{j\ell})\big),
\]
where the $\sigma_j$ describe pairwise distinct full atomic types,
every $\tau_{j\ell}$ is a full atomic type and
$\psi_{j\ell}$ is one of the following.
\begin{enumerate}[(1)]
\item a formula $x\myd y \wedge S(x,y)$ with $\myd\in\calD_-$,
\item a formula $x\myd y \wedge \neg S(x,y)$ with $\myd\in\calD_-$,
\item an atom $x= y$, 
\item a literal $x\not= y$.
\end{enumerate}

For each $j\le K$,  the
disjuncts with formulas $\psi_{j\ell}$ of types (1) and (3) can be combined into one
$\exists$-constraint. Disjuncts with formulas of type (2) can be
treated as follows. If $\sigma_j=\tau_{j,\ell}$ then the $j$-th
conjunct can be completely removed as it becomes a tautology.
If $\sigma_j\not=\tau_{j,\ell}$, the disjunct can be deleted as it can
not be satisfied. If this results in an empty disjunction then
$\sigma_j$ is deleted from
$\Sigma$.
The numbers $K$ and $M$ are at most exponential in
$|\varphi|$.

Altogether, we obtain an instance $L=(\Sigma,C_\exists,C_\forall)$ of $\LPP$
such that a $\Sigma$-labeled point set is a solution to $L$ if and only if
the corresponding labeled point set  $M$ is a model of $\varphi$.
Furthermore, the size of $L$ is at most exponential in $|\varphi|$ and
the construction of $L$ from $\varphi$ can be done in exponential time.
 \end{myproof}


\subsection{$\LPP$ is in $\PSPACE$}

In order to prove the upper bound of Theorem \ref{theo:twodcompl} it would be sufficient
to give a polynomial space algorithm that decides whether a given
\LPP-instance has a solution where no two points are on the same
vertical. However, there is an algorithm that decides the existence of
{\em any} finite solution. 

\begin{proposition}\label{prop:lpp}
  Whether an instance $L$ of $\LPP$ has a finite solution can be
decided in polynomial space. 
\end{proposition}

The algorithm follows a plane sweep approach, that is, it guesses a
solution for $L$ (horizontal) line by line. The crucial point is to
show that this is possible in polynomial space.

With 
every line $l$ of a set of labeled points we associate a
\textit{profile} that, roughly speaking, contains all neccessary
information to validate the constraints from $L$ for points on line
$l$. Profiles respecting the constraints from $L$ will be called
\textit{valid}.  We will prove that a solution of $L$ can be
constructed from a sequence of valid profiles where every pair of
successive profiles is \textit{consistent}. Then, the algorithm guesses a sequence of profiles and verifies
validity and consistency locally. Since the size of each profile is
polynomial in the size of $L$, this is a polynomial space algorithm
proving Proposition \ref{prop:lpp}.

%

  Now, we fill in the details of the proof, and start by giving some definitions. 
  \subsubsection*{Profile of a Horizontal Line.} 

A \textit{horizontal line} of a set $\calP$ of labeled points is a
number $r\in\nat$ such that $\calP$ contains points $p$ and $q$ with
$p.y\le r\le q.y$. Intuitively, a horizontal line is just a line that
lies between the top-most and the bottom-most point of $\calP$. We say
that a point from $\calP$ \textit{lies on $r$} if $p.y=r$.
We associate with every horizontal line $r$ of a set $\calP$ a profile
$\Pro(r)$ that represents information about the points on $r$ and on
other points of $\calP$ that are ``important'' with respect to
$\exists$-constraints and $\forall$-constraints.  Intuitively, for
each symbol and each ``direction'' the profile contains the $x$-value
of the leftmost and the rightmost point in that direction. 

Let $D_P :=
  \{\pe\}\cup (\{\pup, \pdown\} \times \{\S, \nS\}) $ denote the set
  of \textit{profile constraints}. 
For two rows $r,s$ we denote by
$d(r,s)$ the profile constraint of $s$ relative to $r$, that is
\[
d(r,s)=
\begin{cases}
  (\dn,\nS) & \text{if $s>r+1$},\\
  (\dn,\S) & \text{if $s=r+1$},\\
  \pe & \text{if $s=r$},\\
  (\ds,\S) & \text{if $s=r-1$},\\
  (\ds,\nS) & \text{if $s<r-1$}.
\end{cases}
\]

A {\em profile} is a sequence $P=A_1
  \ldots A_k$ of non-empty sets of  pairs from $\Sigma \times D_P$,
  such that
  \begin{iteMize}{$\bullet$}
    \item every set $A_i$ contains at most one pair of the form $(\sigma,\pe)$,
    \item for each $\sigma$ there are at most two indices $i,j$ such
      that $(\sigma,\pe)\in A_i$ and $(\sigma,\pe)\in A_j$, and
    \item for every $i$ there is some $(\sigma, d)\in A_i$ such that
      $i$ is either the minimal or the maximal index  of a set containing $(\sigma, d)$.
  \end{iteMize}
The latter two conditions ensure that the number of sets in a  profile
is at most $10 |\Sigma|$. 
   A profile is a \textit{top profile}, if none of its sets contains
   any pair $(\sigma, (\pup, \S))$ or  $(\sigma, (\pup,
   \nS))$. Analogously a \textit{bottom profile} is a profile that
   does not contain any pairs $(\sigma, (\pdown, \S))$ or  \mbox{$(\sigma, (\pdown, \nS))$}. 

Let, for the following, a finite labeled point set $\calP\subs\nat^2$
be fixed. Let $\calP'$ be the set of all points $p$ from $\calP$ for which $p.x$ is minimal or
maximal within all points with label $p.l$ in row $p.y$.  More precisely let
\begin{eqnarray*}
  \begin{split}
    \calP'=&\{p\in \calP\mid \neg\exists q\in\calP: q.l=p.l, q.y=p.y,
    q.x<p.x\}\\
& \cup \{p\in \calP\mid \neg\exists q\in\calP: q.l=p.l, q.y=p.y,
    q.x>p.x\}\\
  \end{split}
\end{eqnarray*}
As $\calP'$ has, for each label $\sigma$, at most two $\sigma$-labeled
points per row, it has at most $2|\Sigma|$ points per row altogether.

We describe next, how the profile $\Pro(r)$ is constructed, for every
horizontal row $r$ of $\calP'$. Let thus, a row $r$ be fixed in the following.
For each
$\sigma\in\Sigma$ and every profile constraint $d\in D_P$, we
choose a point $q=q(\sigma,d,\min,r)\in\calP'$ with minimal
value $q.x$ such that $d(r,q.y)=d$ and a point $q=q(\sigma,d,\max,r)\in\calP'$ with maximal
value $q.x$ such that $d(r,q.y)=d$.
Note that  $q(\sigma,d,\min,r)$ and  $q(\sigma,d,\max,r)$ can be equal and do not necessarily exist for every $\sigma$, $d$ and $r$.

Let $A'_{r,1},\ldots,A'_{r,k_r}$ be the equivalence classes of points
of the form $q=q(\sigma,d,\min,r)$ and $q=q(\sigma,d,\max,r)$ with
arbitrary $\sigma$ and $d$ and
the same value $q.x$, ordered by $q.x$. For every $i\le k_r$, we denote by $A_{r,i}'.x$  the value of the $x$-coordinate of the
points in $A_{r,i}'$. 

Even though the classes in $A'_{r,1},\ldots,A'_{r,k_r}$ contain
  all points that are relevant for the validity of the given
  constraints, it might be necessary to add some more points for
  technical reasons, as follows.

For each $i\le k_r$, let the set $A''_{r,i}$ be defined as follows.

\[
  \begin{split}
    A''_{r,i}  =  A'_{r,i} & \cup \{q\in \calP'\mid \forall s
    (r<s\le q.y \rightarrow \exists j\, A'_{s,j}.x=A'_{r,i}.x)\}\\
 & \cup \{q\in \calP'\mid \forall s
    (r>s\ge q.y \rightarrow \exists j\, A'_{s,j}.x=A'_{r,i}.x)\}
  \end{split}
\]

For each $i\le k_r$ we define the set $A_{r,i}$ as
\[
A_{r,i}=\{(q.l,d(r,q)) \mid q\in A''_{r,i}\}
\]
and, finally, we let
\[
\Pro(r)=A_{r,1},\ldots,A_{r,k_r}
\]

For every $i\le k_r$, we call $A_{r,i}'.x$ the \textit{$x$-value of} position $i$ in $\Pro(r)$.

It is easy to see that $\Pro(r)$ is indeed a profile as defined
above. We refer to Figure \ref{fig:lpp1} for an example profile extraction. 

Let $y_{\max}$ be the maximal $y$-value of a point in $\calP'$ and
$y_{\min}$ be the minimal $y$-value of a point in $\calP'$.  
    The \textit{profile sequence $\Seq(\calP)$} of $\calP$ is the
    sequence $\Pro(y_{\min})\Pro(y_{\min}+1)\ldots\Pro(y_{\max})$. 

  \newcommand{\addpoint}[5]{
    \cnode*(#1,#2){\pointsize}{1}
    \nput[labelsep=0.5mm]{315}{1}{\scriptsize #3}
    \psline[linestyle=dotted]{-}(#1,#4)(#1, #5)
  }
  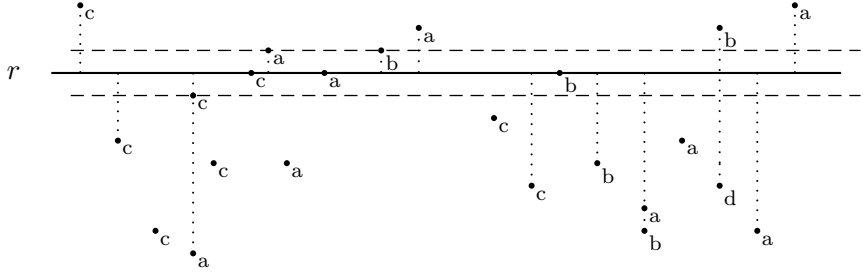
\begin{figure}[t] 
\centering
	\psset{xunit=0.5cm}
	\psset{yunit=0.3cm}
      \begin{pspicture}(22,10.5)
	\newdimen\pointsize
	\pointsize=0.4mm
	\newcount\ul 
	\ul=8 
	\newcount\ll 
	\ll=4
	

	\psline{-}(1,\ul)(22,\ul)


        \rput(0,\ul){$r$}

	\addpoint{7}{\ll}{a}{\ll}{\ll}
	\addpoint{8}{\ul}{a}{\ul}{\ul}

	\addpoint{5.8}{\ul}{c}{\ul}{\ul}
	\addpoint{14}{\ul}{b}{\ul}{\ul}
	\addpoint{4.8}{\ll}{c}{\ll}{\ll}
	\addpoint{15}{\ll}{b}{\ul}{\ll}

\psline[linestyle=dashed,linewidth=0.5pt](1,7)(22,7)
\psline[linestyle=dashed,linewidth=0.5pt](1,9)(22,9)

	\addpoint{3}{1}{c}{1}{1}
	\addpoint{4}{0}{a}{0}{\ul}
	\addpoint{13}{3}{c}{3}{\ul}
	\addpoint{16}{2}{a}{2}{2}
	\addpoint{16}{1}{b}{1}{\ul}
	\addpoint{18}{3}{d}{3}{\ll}
	\addpoint{19}{1}{a}{1}{\ul}
%
	\addpoint{2}{5}{c}{5}{\ul}
	\addpoint{4}{7}{c}{7}{7}
	\addpoint{12}{6}{c}{6}{6}
	\addpoint{17}{5}{a}{\ll}{\ll}
%
	\addpoint{1}{11}{c}{11}{\ul}
	\addpoint{6}{9}{a}{9}{\ul}
	\addpoint{9}{9}{b}{9}{\ul}
	\addpoint{10}{10}{a}{10}{\ul}
	\addpoint{18}{10}{b}{10}{\ll}
	\addpoint{20}{11}{a}{11}{\ul}

      \end{pspicture} 
      \caption{Line with profile $(c,
        \pup)(c, \pdown)\{(a, \pdown), (c, (\pdown,\S))\}(c, \pe)(a, (\pup,S))(a, \pe)(b,
        (\pup,\S))$ $(a, \pup) (c, \pdown)(b, \pe)(b, \pdown)(b, \pdown)\{(b,\pup),(d,\pdown)\}(a,
        \pdown)(a, \pup)$. Dotted vertical lines indicate points contributing
        to the profiles. The dashed horizontal lines indicate the
        predecessor and successor line, respectively. Singleton sets
        in profiles are written without braces.  Profile constraints $(\pup,\neg S)$
        and $(\pdown,\neg S)$ are abbreviated as $\pup$ and $\pdown$,
        respectively. Thus, e.g., the first pair $(c,\pup)$
        abbreviates the full notation $\{(c,(\pup,\neg S))\}$. 
\label{fig:lpp1}}
    \end{figure}

  \subsubsection*{Valid Profiles.} Next we define some necessary
  conditions that a profile of a horizontal line in a solution of an
  \LPP-instance $L$ must fulfill. 

  Let $A = A_1 \ldots A_k$ be a profile. We say that a pair
  $(\sigma, \pe)\in A_i$ {\em fulfills a directional constraint
    $(\circ_d, s_d) \in \calD \times \{\S, \nS\}$ with respect to a pair $(\tau,c)\in A_j$},
  if the following conditions are
  satisfied, where $c=(\circ_c, s_c)$ or $c=\pe$:

  \begin{iteMize}{$\bullet$}
   \item  (Horizontal Conditions) 	
      \begin{iteMize}{$-$}
	\item If $\circ_d \in \{\dnw, \dw, \dsw\}$ then $i > j$.
	\item If $\circ_d \in \{\dne, \de, \dse\}$ then $i < j$.
	\item If $\circ_d \in \{\dw, \de\}$ then $c = \pe$.
	\item  If $\circ_d \not\in \{\dw, \de\}$ then $c\not=\pe$ and $s_d = s_c$.
      \end{iteMize}
    \item  (Vertical Conditions) 	
      \begin{iteMize}{$-$}
	\item If $\circ_d \in \{\dnw, \dn, \dne\}$ then $\circ_c =  \pup$. 
	\item If $\circ_d \in \{\dsw, \ds, \dse\}$ then $\circ_c = \pdown$.
	\item If $\circ_d \in \{\dn, \ds\}$ then $i = j$.
     \end{iteMize}
  \end{iteMize}
 A pair
  $(\sigma, \pe)\in A_i$ {\em fulfills an inequality constraint
    $\neq$ with respect to a pair $(\tau,c)\in A_j$},
   if $i\not=j$ or $c\not=\pe$.

  A profile $A = A_1 \ldots A_k$ is \textit{valid} with respect to an
  existential constraint $(\sigma, E)$ if, for every $i$ and every
  $(\sigma, \pe)\in A_i$, there is some $(\tau,d)
  \in E$ and some $(\tau, c)$ in some $A_j$ such that  $(\sigma, \pe)$
  fulfills  $d$ with respect to $(\tau, c)$.

  The profile $A$ is \textit{valid} with respect to a universal
  constraint $(\sigma, \tau,d)$ if there {\em do not exist any} $i,j$ and pairs $(\sigma, \pe)\in
  A_i$  and $(\tau, c)\in A_j$ such that $(\sigma, \pe)$
  fulfills  $d$ with respect to $(\tau, c)$.

A profile $A$ is \textit{$L$-valid}  if it is valid with respect to
all (existential and universal) constraints from $L$.

\begin{example}
   Let $L$ be  an $\LPP$ instance and $A = A_1 \ldots A_k$  a
   profile. If $L$ contains an $\exists$-constraint
   $(\sigma,\{(\tau_1, (\dnw, \nS)),(\tau_2, \not=),(\tau_3, (\dse,
   \S))\})$ then for any $(\sigma,\pe)$ in $A_i$ there has to be
   \begin{iteMize}{$\bullet$}
   \item an
   occurrence $(\tau_1,(\pup, \nS))$ in $A_j$ for some $j<i$,
 \item an   occurrence of a pair $(\tau_2, c)$ in $A_j$ with $c \neq
   \cdot$ if $j = i$, or
 \item an  occurrence $(\tau_3,(\pdown, \S))$ in $A_{l}$ for some $l>i$. 
   \end{iteMize}
If $L$ contains a $\forall$-constraint $(\sigma, \tau, (\dsw, \S))$
then there should be no $i>j$ with $(\sigma, \pe)\in A_i$ and $(\tau,
(\pdown, \S)) \in A_j$.  
\end{example}
 
%

   \subsubsection*{Consistent Pairs of Profiles.} 
     Let $A = A_1 \ldots A_k$ and $B = B_1 \ldots B_m$ be profiles. The pair $(A, B)$ is said to be \textit{consistent}, if there is a
     function
 \[
s:\big(\{a\}\times\{1,\ldots,k\}\cup\{b\}\times\{1,\ldots,m\}\big)\to\rat
\]
     that is strictly monotone in its second parameter and
fulfills the following conditions (C1) and (C2) for every $\sigma \in
\Sigma$.
\begin{enumerate}[(C1)]
  \item
    \begin{enumerate}[(a)]
    \item If for some $i$, $A_i$ contains $(\sigma, (\pup, \S))$
    then there is some $j$, such that $B_j$
    contains  $(\sigma, \pe)$ and $s(a,i) = s(b,j)$.
  \item If for some $j$,  $B_j$   contains  $(\sigma, \pe)$ then there is
    some $i$ such that $A_i$ contains
    $(\sigma, (\pup, \S))$ and $s(a,i) = s(b,j)$.
    \end{enumerate}
 \item
   \begin{enumerate}[(a)]
   \item If for some $i$,  $A_i$ contains $(\sigma, (\pup, \nS))$, there
   is some $j$ such that $B_j$    contains $(\sigma,  (\pup, \S))$  or $(\sigma,  (\pup, \nS))$ and
    $s(a,i) = s(b,j)$.
  \item  If for some $j$,  $B_j$ contains $(\sigma,
    (\pup, \S))$  or $(\sigma,  (\pup, \nS))$ and $j$ is minimal or
    maximal with this property then  there is some $i$ with $s(a,i) = s(b,j)$ such
    that $A_i$ contains $(\sigma, (\pup, \nS))$.
  \item If for some $j$,  $B_j$ contains $(\sigma,
    (\pup, \S))$  or $(\sigma,  (\pup, \nS))$ and there is some $i$
    with $s(a,i) = s(b,j)$ then $(\sigma, (\pup, \nS))\in A_i$.
   \end{enumerate}
\end{enumerate}
 The analogous conditions hold with respect to
  occurrences of $(\sigma,(\pdown,\S))$ and $(\sigma,(\pdown,\nS))$
  \mbox{in $B$}.

 Intuitively, $s$ ``synchronizes'' points in $A$ with points
in $B$ by mapping them to the same ``$x$-coordinate''.  
The conditions (C1) ensure that successive roles fit together with
respect to profile constraints of the form $(\dn,S)$ and
$(\ds,S)$. More precisely, if $(\sigma,(\dn,S))$ occurs in $A$ then $B$
should have a $\sigma$-point at the same vertical position. 
The conditions of (C2) ensure consistency for profile constraints of the form $(\dn,\nS)$ and
$(\ds,\nS)$. (C2a) ensures that the requirement that there is a
$\sigma$-point above a certain position of $A$ is reflected in
$B$. (C2b) ensures that $A$ contains the full information about
minimal and maximal points above $B$. Finally, (C2c) guarantees that
if a point $u$ for $A$ is synchronized with a point $v$ for $B$ then
its profile
carries the full information about the points above $v$. 

If $s(a,i)=s(b,j)$ we say that position $i$ of $A$ and position $j$ of
$B$ are \textit{synchronized}.  We additionally require for $s$ that
unless requested by (C1) or (C2) (or their downward counterparts) $s(a,i)\not=s(b,j)$ holds for
every $i$ and $j$. This requirement
guarantees that, even though $s$ is not uniquely determined, the set
of pairs of synchronized positions is unique. 


A sequence $\calS = P_1 \ldots P_N$  of profiles is called \textit{consistent}, if $P_1$ is a bottom profile, $P_N$ is a top profile and $(P_i, P_{i+1})$ is consistent for every $i \in \{1, \ldots, N-1\}$.

If $\calS = P_1 \ldots P_N$ is a consistent sequence of profiles, we say that a position $k$ of $P_i$ is \textit{connected to} position
  $\ell$ of $P_j$, where $j>i$ if
  \begin{enumerate}
  \item $j=i+1$ and the two positions are synchronized or
  \item position $k$ in $P_i$ is synchronized with some position in $P_{i+1}$
    which in turn is connected to position $\ell$ in $P_j$.
  \end{enumerate}

 The following proposition characterizes the \LPP-instances with
 finite solutions in terms of existence of consistent sequences of
 valid profiles and is key to our algorithm.

  \begin{proposition}\label{prop:modelcharacterisation}
     Let $L$ be an instance of \LPP. Then $L$ has a finite
     solution if and only if there is a finite consistent sequence of $L$-valid
     profiles. 
  \end{proposition}

  \begin{proof}
    Let $\calP$ be a model of $L$. We claim that $\Seq(\calP) = P_1
    \ldots P_N$ is a consistent sequence of $L$-valid profiles. Let
    $\calP'$ be defined as in the definition of $\Seq(\calP)$.
To this
    end, let $k\le N$ be arbitrarily chosen such that
    $P=P_k=A_1,\ldots,A_\ell$, for some $\ell$.  We show that $P$ is
    valid with respect to all existential and universal constraints
    from $L$. 

Let first $(\sigma,E)$ be an $\exists$-constraint from $L$ and let
      $(\sigma,\pe)\in A_i$, for some $i\le \ell$. Let $r$ be a
      horizontal line of $\calP$ that corresponds to $P_k$ and let $p$
      be the point on $r$ that corresponds to  $(\sigma,\pe)\in
      A_i$. As $\calP$ is a solution for $L$ there is some
      $(\tau,d)\in E$ and 
a point  $q\in\calP$ with label $\tau$ such that $(p,q)$
      satisfies $d$. We first consider the case that $d$ is a
      directional constraint  $(\myd,s_d)$.
      \begin{iteMize}{$\bullet$}
      \item If $\myd=\dne$ then $q$ is to the northeast of $p$, in
        particular, $p.y<q.y$ and $p.x<q.x$. 
Let $q'$ be
        $q(\tau,(\dn,s_d),\max,r)$. By definition of $q(\tau,(\dn,s_d),\max,r)$, it
        holds $p.x<q.x \leq q'.x$
and  $p.y+1 < q'.y$  in case $s_d = \nS$ (in case, $s_d=\S$, we get
  $p.y+1 = q'.y$ instead). Therefore,
        $(\tau,(\dn,s_d))$ occurs in $A_j$, for some $j>i$. Clearly,
        $(\sigma,\pe)$ and  $(\tau,(\dn,s_d))$ are consistent with
        respect to $d$. Analogous arguments can be applied if
        $\myd\in\{\dnw,\dsw,\dse\}$. 
      \item  If $\myd=\de$ then $q$ is to the right (east) of $p$, in
        particular, it holds $p.y=q.y$ and $p.x<q.x$. Without loss of
        generality, $q=q(\tau,\pe,\max,r)$. By definition of
        $q(\tau,\pe,\max,r)$, it 
        holds $p.y=q.y$ and $p.x<q.x$. Therefore,
        $(\tau,\pe)$ occurs in $A_j$, for some $j>i$. Clearly,
        $(\sigma,\pe)$ and  $(\tau,\pe)$ are consistent with
        respect to $d$. An analogous argument can be applied if
        $\myd=\dw$.
      \end{iteMize}
If $d$ is an inequality constraint, then $q\not=p$ but the relative
position of $p$ and $q$ can be arbitrary. If $p.y<q.y$ we can argue as
above with $q'=q(\tau,(\dn,\nS),\max,r)$ or
$q'=q(\tau,(\dn,\S),\max,r)$ (if the former does not exist). Note
however, that the corresponding pair
$(\tau,(\dn,\S))$ or $(\tau,(\dn,\nS))$ could be in $A_i$. However, as
the second component of the pair is not $\pe$ this is sufficient. If
$p.y>q.y$ we can argue analogously. If $p.y=q.y$ and $p.x<q.x$ we can
argue as above in the case $\myd=\de$. Likewise,  we can
argue as above in the case $\myd=\dw$
if $p.y=q.y$ and $p.x>q.x$. 
Altogether we can conclude that $P$ is valid with respect to $(\sigma,E)$.

Let now $(\sigma,\tau,d)$ be a $\forall$-constraint from $L$ and let $(\sigma, \pe)\in
  A_i$,   for some $i$, $p=q(\sigma,\pe,m,r)$ for some $m\in\{\min,\max\}$ the
  corresponding point and $(\tau, c)\in A_j$ for some $j$. 
Towards a contradiction let us now assume that $(\sigma, \pe)\in
  A_i$  and $(\tau, c)\in A_j$ are consistent with respect to
  $d=(\myd,s_d)$. We distinguish cases based on $\myd$.
  \begin{iteMize}{$\bullet$}
  \item We first consider the case $\myd=\dne$. Thus, we can conclude that $i<j$,
    $c\not=\pe$. Hence  $c=(\myc,s_c)$ for some $\myc,s_c$ and we can
    further conclude $s_d=s_c$ and $\myc=\dn$. Let $q\in\calP'$ be the
    point corresponding to   $(\tau, c)\in A_j$. Clearly, $q.l=\tau$
    and $(p,q)$ satisfy $d$ and therefore $\calP$ is not
    a solution for $L$, the desired contradiction. For
    $\myd\in\{\dnw,\dsw,\dse\}$ we can argue analogously.
  \item We next consider the case $\myd=\de$. Thus, we can conclude
    that $i<j$ and $c=\pe$.  But
    then it is easy to see that $p\de q$, where $q$ is the point
    corresponding to $(\tau,\pe)$ in $A_j$ and therefore $\calP$ is not
    a solution for $L$, the desired contradiction. For
    $\myd=\dw$ we can argue analogously.
  \item Finally, we consider the case $\myd=\dn$. Thus, we can
    conclude  $c\not=\pe$ and therefore  $c=(\myc,s_c)$ for some
    $\myc,s_c$. Furthermore, $\myc=\dn$, $s_c=s_d$ and $i=j$.  But
    then it is easy to see that $p\dn q$, where $q$ is the point
    corresponding to $(\tau,(\dn,s_c))$ in $A_j$ and therefore $\calP$ is not
    a solution for $L$, the desired contradiction. For
    $\myd=\ds$ we can argue analogously.
  \end{iteMize}
Altogether we can conclude that $P$ is valid with respect to
$(\sigma,\tau,d)$ and therefore $P$ and all profiles in $\Seq(\calP)$ are
$L$-valid. \\

We continue by showing that $\Seq(\calP)$ is consistent. Clearly,
$P_1$ is a bottom profile and $P_N$ is a top profile as there are no
points in $\calP'$ below the row corresponding to $P_1$ or above the
row  corresponding to $P_N$. It thus remains to be shown that, for
every $k<N$, $P_k$ and $P_{k+1}$ are consistent. 

  Let $P_k = A_1 \ldots A_\ell$ and $P_{k+1} = B_1 \ldots B_m$ be
  profiles and let $r$ and $r+1$ be the horizontal rows of $P_k$ and
  $P_{k+1}$, respectively. 
 
We define the synchronizing function 
\[
s:\big(\{a\}\times\{1,\ldots,\ell\}\cup\{b\}\times\{1,\ldots,m\}\big)\to\rat
\]
as follows. For $u\le \ell$, we let $s(a,u)$ be the $x$-value of
$A_u$ and for $v\le m$, we let $s(b,v)$ be the $x$-value of
$B_v$. We need to show conditions (C1) and (C2). 

Towards (C1a), let $(\sigma,(\dn,\S))\in A_i$, for some $i$. 
By construction of $P_k$, there is a $\sigma$-labeled point
$q\in\calP'$ with $q.y=r+1$ and $q.x$ is the $x$-value of position $i$
in $P_k$. Therefore, there is a $j$ such that $(\sigma,\pe)\in B_j$
and $s(a,i)=s(b,j)$. Towards (C1b), let  $(\sigma,\pe)\in B_j$ and let
$q\in\calP'$ be its corresponding point. As there are at most two
$\sigma$-labeled points on $r+1$, $q=(\sigma,(\dn,\S),m,r)$, for some
$m\in\{\min,\max\}$. Therefore, by definition of $\Pro(r)$, there is
$i$ such that  $(\sigma,(\dn,\S))\in A_i$ and $s(a,i)=s(b,j)$. 
Analogously, for $(\sigma,(\ds,\S))\in A_i$.

Towards (C2a),  let $(\sigma,(\dn,\nS))\in A_i$, for some $i$. Let
$q=q(\sigma,(\dn,\nS),m,r)$ be the corresponding point in $\calP'$, for some
$m\in\{\min,\max\}$. By definition of $\Pro(r)$, either
$q=q(\sigma,(\dn,\nS),m,r+1)$ or $q=q(\sigma,(\dn,\S),m,r+1)$. In either
case (C2a) holds, by definition of  $\Pro(r+1)$.

In order to show (C2b), let $j$ be minimal such that $B_j$ contains
$(\sigma,(\dn,\S))$ or $(\sigma,(\dn,\nS))$. Therefore,
the point $q$ corresponding to $B_j$ is a leftmost point above $r+1$ with label
$\sigma$. Thus, $q=q(\sigma,(\dn,\nS),\min,r)$ and there is some $i$
with $s(a,i)=s(b,j)$ and $q$ induces a pair $(\sigma,(\dn,\nS))$ in $A_i$. 

Finally, (C2c) is guaranteed by the definition of $\Pro(r)$. 

Altogether, we can conclude that $\Seq(\calP)$ is consistent.

This concludes the first part of the proof that is, if $L$ has a finite
     solution then there is a finite consistent sequence of $L$-valid
     profiles.

    Now, let  $\calS = P_1 \ldots P_N$ be a consistent sequence of $L$-valid
    profiles. There are sets $A_{k,i}$ and numbers $m_k$ such that for
    each $k\le 1$, $P_k=A_{k,1},\ldots,A_{k,m_k}$. 

We construct a finite set $\calP$  of labeled points with
    $\Seq(\calP) = \calS$ and prove that $\calP$ is a model
    of $L$. Intuitively, the idea is that for every occurrence of $(\sigma, \pe)$
    in profile $P_i$ the $i$th line of $\calP$ contains a point $p$
    labeled with $\sigma$. That is, the points in $\calP$ only have
    $y$-values from $\{1,\ldots,N\}$. During the construction of
    $\calP$ we will use rational, possibly non-integer values for
    $x$-coordinates. We note that if there is a finite
    solution for $L$ with rational $x$-values, then there is also a finite
   solution for $L$ with natural $x$-values which can be obtained by
   replacing, for every $k$, the $k$-th smallest $x$-value by
   the number $k$. Clearly, this transformation does not affect
   the successor relation on the $y$-values.

More formally, we assign
    to every set of a profile a point in $\rat\times\nat$. To this
    end, we define, for every $k\in\{1,\ldots,n\}$,
 a function $\pi_k:\{1,\ldots,m_k\}\to\rat$
    with the interpretation that the
    point $(\pi_k(i),k)$ is assigned to the set $A_{k,i}$. Whenever a set $A_{k,i}$ contains a pair $(\sigma,
\pe)$ we add a point $p$  to $\calP$ with $p.y=k$, $p.x=\pi_k(i)$ and
$p.l=\sigma$. Therefore, $\calP$ is defined as soon as the functions
$\pi_k$ are defined.

The function $\pi_1$ for the first row is defined by $\pi_1(i)=i$, for
every $i\le m_1$. Let us now assume that the function
$\pi_1,\ldots,\pi_{k-1}$ are already defined. Since $(P_{k-1},P_k)$ is
consistent there is a synchronizing function $s$ for them.  For every
pair $(i,j)$ such that $A_{k-1,i}$ and $A_{k,j}$ are synchronized by
$s$ we define $\pi_k(j)=\pi_{k-1}(i)$. 
Let $j_0,j_1$ be such that $\pi_k(j_0)$ and $\pi_k(j_1)$ are already
defined by the previous step, and  $\pi_k(j)$ is not defined for every $j$ with $j_0 < j < j_1$. We define $\pi_k(j)$ for all $j$,
$j_0<j<j_1$ by picking values such that
\begin{iteMize}{$\bullet$}
\item $\pi_k(j_0)<\pi_k(j)<\pi_k(j')<\pi_k(j_1)$, for every pair $j,j'$ with
  $j_0<j<j'<j_1$, and
\item for every $j$, $j_0<j<j_1$, $\pi_k(j)$ is not in the range of
  $\pi_\ell$ for any $\ell<k$.
\end{iteMize}
As the interval $(\pi_k(j_0),\pi_k(j_1))$ contains infinitely many
points from $\rat$ such values can be found. We note that this
  choice of $x$-values ensures that two points can be on the same
  vertical only if they are connected.
This concludes the construction of $\calP$.

It remains to prove that $\calP$ is indeed a solution for $L$. 

To this end, we first prove the following claim.
\begin{myclaim}\label{claim:calP}
  Let $k\le N$ and $i\le m_k$.
  \begin{enumerate}[(a)]
  \item There is a $\tau$-labeled point
  $q\in\calP$ with $q.x=\pi_k(i)$, $q.y>k+1$ and $q.l=\tau$ if and only
  if $(\tau,(\dn,\nS)\in
  A_{k,i}$.
  \item There is a $\tau$-labeled point
  $q\in\calP$ with $q.x=\pi_k(i)$, $q.y=k+1$ and $q.l=\tau$ if and only
  if $(\tau,(\dn,\S)\in A_{k,i}$.
  \end{enumerate}
Furthermore, the
  corresponding statement with $\ds$ in place of $\dn$ hold as well.
\end{myclaim}
\begin{proofof}{Claim \ref{claim:calP}}
  Let us first assume that  $(\tau,(\dn,\nS))\in A_{k,i}$. By consistency of the profiles and
  induction it follows that for some $\ell > k+1$ and $j\le m_\ell$,
  $(\tau,\pe)\in A_{\ell,j}$ and position $j$ of $P_\ell$ is connected
  to position $i$ of $P_k$. By definition of $\pi$ and induction we can
  conclude that $\pi_k(i)=\pi_\ell(j)$ and therefore, by construction
  of $\calP$ from $\pi$, such a point $q$ indeed exists. If
  $(\tau,(\dn,\S))\in A_{k,i}$ an analogous argument yields
  $(\tau,\pe)\in A_{k+1,j}$ and again a point $q$ with the desired
  properties. Thus, we can conclude that the ``if''-statements in (a)
  and (b) hold. The  corresponding conclusion for $\ds$ in place of $\dn$ naturally also holds.

Let us now assume that $q$ is as in the statement of the claim. By
construction of $\calP$ there must be $\ell$ and $j$ such that
$(\tau,\pe)\in A_{\ell,j}$, $\ell>k+1$, and $\pi_\ell(j)=\pi_k(i)$. By
definition of $\pi$ and induction it follows that position $j$ of
$P_\ell$ and position $i$ of $P_k$ are connected. By consistency
condition (C1) and the definition of $\pi$ we can conclude that $(\tau,(\dn,\S))\in
A_{\ell-1,j'}$, for some $j'\le m_{\ell-1}$ with
$\pi_{\ell-1}(j')=\pi_\ell(j)$. This yields the ``only if''-part of (b).
By consistency condition (C2), definition of $\pi$ and induction we can conclude that  for every
$u<\ell-1$ there is some $v\le m_u$ such that $(\tau,(\dn,\nS))\in
A_{u,v}$ and $\pi_u(v)=\pi_\ell(j)$. Thus, in particular, either
$k=\ell-1$ and $(\tau,(\dn,\S))\in
A_{k,i}$ or $k<\ell-1$ and $(\tau,(\dn,\nS))\in
A_{k,i}$. This yields the ``only if''-part of (a). The
  corresponding conclusion for $\ds$ in place of $\dn$ holds
  again. This completes the proof of Claim \ref{claim:calP}. 
\end{proofof}

Now we can show that $\calP$ is a solution for $L$. 

Let first
$(\sigma,E)$ be an $\exists$-constraint from $L$ and $p\in\calP$ with $p.l = \sigma$. Let
$k$ be the number of the corresponding profile and $i$ the position of
the corresponding set in $P_k$. By construction of $\calP$,
$(\sigma,\pe)\in A_{k,i}$. As $\calS$ is $L$-valid, there is some
$(\tau,d)\in E$, some $j\le m_k$ and some $(\tau,c)\in A_{k,j}$ such
that $(\sigma,\pe)$ and $(\tau,c)$ are consistent with respect to
$d$.
\begin{iteMize}{$\bullet$}
\item Let us first assume that $d=(\dne,s_d)$, for some $s_d$. 
From statements (a) and (b) of  Claim \ref{claim:calP} it follows that
there exists a $\tau$-labeled point $q$ corresponding to $(\tau,c)$ in
$ A_{k,j}$ such that $(p,q)$ satisfies $d$. For directions
$\dnw,\dsw,\dse$ we can argue analogously.
\item Let us now assume that $d=(\de,\nS)$. By validity, $c=\pe$ and
  thus there is a $\tau$-labeled point $q$ with $q.y=k$ and
  $q.x=\pi_k(j)>p.x$. Likewise, for $\dw$.
\item Let finally $d=\not=$. From validity it follows that $i\not=j$
  or $c\not=\pe$. 
Similarly as in the two previous cases,
  we can conclude that there is a 
  $\tau$-labeled point $q$ in the row of $p$ or somewhere else,
  possibly on the same vertical line as $p$ but different from $p$. 
\end{iteMize}

Let now $(\sigma,\tau,d)$ be a $\forall$-constraint. For the sake of a
contradiction, let us assume that there are $p$ and $q$ that do
\textit{not} satisfy this constraint. By definition this means that
$p.l=\sigma$, $q.l=\tau$ and $(p,q)$ satisfies $d$. Let $A_{k,i}$ and
$A_{\ell,j}$ be the sets from which $p$ and $q$ were obtained. 
\begin{iteMize}{$\bullet$}
\item Let us first assume that $d=(\dne,\S)$. Thus, $\ell=k+1$. As
  $(\tau,\pe)\in A_{\ell,j}$ and $q.x>p.x$, (C1) implies that $(\tau,(\dn,\S))\in
  A_{k,j'}$ for some $j'>i$ contradicting validity. An analogous
  argument applies to the directions $\dnw, \dsw,\dse$ in place of $\dne$.
\item  Let us next assume that $d=(\dne,\nS)$. Thus,
  $\ell>k+1$. Using consistency it can be shown by induction that, for
  every $u< \ell-1$ there is some $v$ such that $(\tau,(\dn,\nS))\in
  A_{u,v}$ and $\pi_u(v)\ge \pi_\ell(j)$. In particular, there is some
  $j'$ such that $(\tau,(\dn,\nS))\in
  A_{k,j'}$ and $\pi_k(j')\ge \pi_\ell(j)>\pi_k(i)$, contradicting validity.
An analogous
  argument applies to the directions $\dnw, \dsw,\dse$ in place of $\dne$.
\item Let us next assume that $d=(\de,\nS)$. Thus, $\ell=k$ and
  $(\tau,\pe)\in A_{k,j}$, an immediate contradiction to
  validity. Likewise for $\dw$ in place of $\de$.
\item Let us finally assume that $d=(\dn,s_d)$, for some $s_d$. By
  Claim \ref{claim:calP} it follows that  $(\tau,(\dn,s_d))\in
  A_{k,i}$, an immediate contradiction to validity.
\end{iteMize}
As all possible cases yield a contradiction we can conclude that
violations of $\forall$-constraints do not occur. Altogether we have
shown that $\calP$ is indeed a solution for $L$.
  \end{proof}

It should be noted that the proof implicitly shows that if an
\LPP-problem has a finite solution then it has a solution with at most
$2|\Sigma|$ points per horizontal line.

%

Now we are ready to complete the proof of Proposition \ref{prop:lpp}.

\begin{proofof}{Proposition \ref{prop:lpp}}
Proposition \ref{prop:modelcharacterisation} allows for testing
satisfiability of a labeled point problem $L$ by checking whether
there is a consistent sequence of $L$-valid profiles. We note that there
is only an exponential number $M$ of profiles, thus if there is a such
sequence of profiles, there is also a sequence of length $\leq M$.  

Whether such a sequence exists can be tested by a non-deterministic
algorithm with polynomial space. The algorithm guesses a sequence
$P_1,\ldots,P_N$ of profiles, for some $N\le M$. It checks that every
profile is $L$-valid and that the sequence is consistent. To this end,
it only needs to store two profiles $P_k,P_{k+1}$ at any time.
Clearly, $L$-validity of a given profile can be tested in polynomial
time, likewise consistency of two profiles $P_k,P_{k+1}$ and whether
$P_1$ is a bottom profile and $P_N$ a top profile.

By Savitch's Theorem the problem can therefore be solved in polynomial space.
\end{proofof}

The proof of Theorem \ref{theo:twodcompl} can now be given by a simple
combination of the results in this section.

\begin{proofof}{the upper bound of Theorem \ref{theo:twodcompl}}
   Let $\varphi$ be a $\FO^2(\leq_1, [\precsim_2, \N_2])$-sentence. By
   Proposition \ref{prop:topoints}, a semi-positive 
  $\FO^2(\calD_-, I_y)$-sentence $\varphi'$ that is equivalent to
  $\varphi$  with respect to finite satisfiability can be computed in
  polynomial time. From $\varphi'$ an equivalent exponential size
  labeled point problem $L$ can be obtained, by Proposition
  \ref{prop:tolpp}. By Proposition \ref{prop:lpp}, finite
  satisfiability of $L$ can be tested in polynomial space. Hence,
  testing finite satisfiability of a $\FO^2(\leq,_1,[\precsim_2,
  \N_2])$-sentence can be done in exponential space. 
\end{proofof}

Already in Section \ref{sec:preliminaries} we saw a strong connection between finite ordered structures and data words. Hence the following theorem follows straightforwardly from Theorem \ref{theo:twodcompl}.

\begin{theorem}
  Finite satisfiability for $\DW^2(\Sigma,\leq_1[\precsim_2, S_2])$ is in $\EXPSPACE$.
\end{theorem}

\section{Lower Bounds for Two-variable Logic on Ordered Structures}\label{sec:lower}

\newcommand{\ECT}{\ensuremath{\text{ExpCorridorTiling}}} 
\subsection{One Linear Order and one Total Preorder}
 \label{sec:hardness}
The following result shows that the upper bound of Theorem
\ref{theo:twodcompl} is sharp. We recall that $\FO^2(\leq_1,
  \precsim_2)$ is an abbreviation for $\FO^2(\FinOrd(\leq_1,
  \precsim_2))$.

\begin{theorem}\label{theo:expspacetwoorders}
  Finite satisfiability for $\FO^2(\leq_1, \precsim_2)$ is \EXPSPACE-hard.
\end{theorem}

\begin{myproof}
In the following, we use the notation $[i,j]$ to denote the set of
natural numbers between $i$ and $j$.

The proof of the theorem is by reduction from the $\EXPSPACE$-complete
problem \ECT. An input to $\ECT$ is a tuple
  $I=(T,H,V,\alpha,\omega,n)$, where
  \begin{iteMize}{$\bullet$}
  \item $T$ is the set of
  \textit{allowed tiles},
\item $V,H\subs T^2$ are the sets of horizontal
  and vertical constraints,
\item $\alpha,\omega\in T$ are tiles for the bottom row and the top row,
  respectively, and
\item $n$ is a natural number, given in binary.
  \end{iteMize}

  A \emph{valid tiling for $I$} is a mapping
  $\lambda:[1,n]\times[1,m]\to T$, for some $m\ge 2$ such that the following constraints  are satisfied:
  \begin{enumerate}[(1)]
    \item the bottom row starts with $\alpha$, that is,
    $\lambda(1,1)=\alpha$;
    \item the top row ends with $\omega$, that is,
    $\lambda(n,m)=\omega$;
    \item all vertical constraints are satisfied, that is, for every
      $i\le n$ and every $j< m$, $(\lambda(i,j),\lambda(i,j+1))\in V$; and,
    \item all horizontal constraints are satisfied,  that is, for every
      $i< n$ and every $j\le m$, $(\lambda(i,j),\lambda(i+1,j))\in H$.
  \end{enumerate}
 
Deciding, whether an instance $I$ for $\ECT$ has a valid tiling is
\EXPSPACE-complete (see, e.g., \cite{Boas1997}). 

We first describe how  tilings for $I$ can be encoded as
structures with a linear order $\le_1$, a total preorder $\precsim_2$ and some
unary relations. Then we describe how an $\FO^2(\le_1,
\precsim_2)$-sentence $\varphi_I$ can be constructed in polynomial time from $I$ that
describes necessary conditions for structures that encode a valid tiling.
Finally, we show that from a model of $\varphi_I$ a
valid tiling for $I$ can be constructed, thus establishing that
$I\mapsto\varphi_I$ is a reduction from $\ECT$ to finite satisfiability
of $\FO^2(\le_1, \precsim_2)$.
 
Let $I=(T,H,V,\alpha,\omega,n)$ be an \ECT-instance and let
$\lambda:[1,n]\times[1,m]\to T$ be a valid tiling for $I$. For
simplicity we assume that $n=2^k$, for some integer $k$. We define a
structure $M(\lambda)$ as follows. The universe $P$ of $M(\lambda)$
consists of points from $\nat\times\nat$. Each position $(i,j)$ of the
tiling is represented in  $M(\lambda)$ by the two elements
\begin{iteMize}{$\bullet$}
\item $p_-(i,j)=(2jn+i,j)$ and
\item $p_+(i,j)=((2j+3)n+i,j)$. 
\end{iteMize}

It should be noted that $x\not=x'$ for any two elements $(x,y)\not=(x',y')$ from $P$.
Therefore, we can define a linear order $\le_1$ on
$P$ by
\[
(x,y)\le_1(x',y') \;\myadef\; x\le x'.
\]
We define the total preorder $\precsim_2$ by
\[
(x,y)\precsim_2(x',y') \;\myadef\; y\le y'.
\]

Thus, the equivalence classes of $\precsim_2$ are just maximal sets of points of $P$ on the
  same horizontal line. 

Furthermore, $M(\lambda)$ has the set $\calC\cup \calS \cup
\{R_-,R_+\}$ of unary relations, where
\begin{iteMize}{$\bullet$}
\item the relations from $\calC=\{C_1,\ldots,C_k\}$ encode the column
  number of elements via  $C_p=\{(i,j)\mid \text{the $p$-th bit of the
  binary representation of $i$ is 1}\}$; 
\item the relations from $\calS=\{S_t\mid t\in T\}$ encode the actual
  tiling via $S_t=\{p_-(i,j), p_+(i,j)\mid \lambda(i,j)=t\}$; and
\item $R_-$ contains all elements of the form $p_-(i,j)$ and $R_+$ all
  elements of the form $p_+(i,j)$. We call the former
  \textit{negative} and the latter \textit{positive} elements. 
\end{iteMize}

See Figure \ref{theo:hardness:fig2} for an outline of the construction.
  \newcommand{\sline}[4]
  {
    \pnode(#2,#1){f1}
    \pnode(#3,#1){f2}
    \ncline{*-*}{f1}{f2}
    \nbput[labelsep=1mm]{#4}
  }

  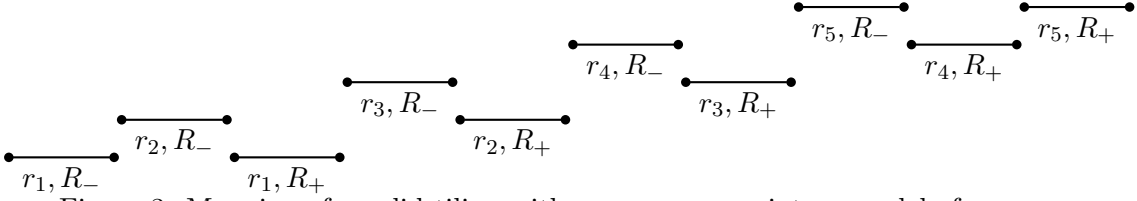
\begin{figure}[t] 
      \psset{xunit=0.5}
      \psset{yunit=0.5}
\centering
    \begin{pspicture}(28,5)
      
      \sline{0}{0.2}{3}{$r_1, R_-$}
      \sline{1}{3.2}{6}{$r_2, R_-$}
      \sline{2}{9.2}{12}{$r_3, R_-$}
      \sline{3}{15.2}{18}{$r_4, R_-$}
      \sline{4}{21.2}{24}{$r_5, R_-$}
      
      \sline{0}{6.2}{9}{$r_1, R_+$}
      \sline{1}{12.2}{15}{$r_2, R_+$}
      \sline{2}{18.2}{21}{$r_3, R_+$}
      \sline{3}{24.2}{27}{$r_4, R_+$}	
      \sline{4}{27.2}{30}{$r_5, R_+$}	
      
    \end{pspicture} 
    \caption{Mapping of a valid tiling with rows $r_1, \ldots, r_5$ into a model of $\varphi$.\label{theo:hardness:fig2}}
  \end{figure}

Next, we describe the construction of $\varphi_I$. As we keep the
intuition of the previous section that elements of a structure with a
linear order $\le_1$ and a total preorder $\precsim_2$ can be
considered as points in the plane, we will refer to the elements of
such a structure as \textit{points} and to the equivalence classes
with respect to $\sim_2$ as \textit{horizontal lines} or simply \textit{lines}\footnote{We
  avoid the term \textit{row} for horizontal lines to avoid confusion
  with the rows of a tiling.}.

Let
$\phit{general}$ be a formula that expresses that
\begin{iteMize}{$\bullet$}
\item every point $p$ is in exactly one of $R_+$ and $R_-$, and
\item for every point $p$ there is exactly one $t$ such that $p\in
  S_t$. In the following, we refer to this $t$ by $t(p)$.
\end{iteMize}

In the following, we associate with every point $p$ a column number
$c(p)$ that is the number whose bit string has a 1 at position $i$ if
and only if $p\in C_i$ (where we count the least significant bit as
position 1). 

We use a formula $\psi_{+1}(x,y)$ to express that the column number of
$y$ is the column number of $x$ plus one:
\begin{equation*}
  \begin{split}
    \psi_{+1}(x,y) = & \bigvee_{i=1}^k \bigl(\neg C_i(x) \land C_i(y) \land\\
& \bigwedge_{j=1}^{i-1} C_j(x) \land \neg C_j(y) \land\\
 & \bigwedge_{j=i+1}^{k} C_j(x) \leftrightarrow C_j(y)\bigr)\\
  \end{split}
\end{equation*}

The next formula $\phit{line}$ expresses that the points of every
horizontal line are as intended. 
More precisely,
$\phit{line}$ expresses that  for every point $p$
\begin{iteMize}{$\bullet$}
\item there is a point $q\in R_-$ with column number $0$ in
  the same line such that $q\le_1 p$;
\item there is a point $q\in R_+$ with column number 0 in
  the same line such that 
($q \le_1 p \Llr p\in R_+$);
\item  there is a point $q$ in the same line with
  $c(q)=c(p)$, $t(q)=t(p)$ and ($p\in R_- \Llr q\in R_+$);
\item  unless $c(p)=n-1$, there is a point $q$
  in the same line with $c(p)+1=c(q)$, $p<_1 q$, ($p\in R_- \Llr q\in R_-$) and
  $(t(p),t(q))\in H$; and
\item there is no point $q\not=p$ in the same line with the same
  column number and ($p\in R_- \Llr q\in R_-$). 
\end{iteMize}

The last formula $\phit{next-line}$ expresses that for every line $l$ which
is not a top line, there is another line $l'$ such that the $R_-$-part
of $l'$ is to the northwest of the $R_+$-part of $l$ and such that the
points in $l$ and $l'$ are compatible with respect to $V$. That is,  $\phit{next-line}$ expresses that
\begin{iteMize}{$\bullet$}
\item  for every point $p\in R_+$ with $c(p)=0$ there is a point $q$
  such that
  \begin{iteMize}{$\bullet$}
  \item $q$ is in the same line as $p$, $c(q)=n-1$ and $t(q)=\omega$,
    or
  \item $q<_1p$, $p\prec_2 q$ and $c(q)=n-1$; and
  \end{iteMize}
\item  for every point $p\in R_+$ and every point $q\in R_-$ with
  $q<_1p$, $p\prec_2 q$ and $c(q)=c(p)$ it holds $(t(p),t(q))\in V$.
\end{iteMize}

The formula $\phit{start}$ expresses that the smallest point with respect to $\leq_1$ has the tile $\alpha$.

Finally, 
\[
\varphi\mydef \phit{general} \land \phit{start} \land \phit{line}
\land \phit{next-line}. 
\]

It is easy to check that $M(\lambda)\models\varphi$ if $\lambda$ is a
valid tiling. 

It remains to show that from every finite model $M$ of $\varphi$ we can construct
a valid tiling $\lambda_M$ for $I$. Let thus $M$ be a finite model of
$\varphi$. The fact that $M\models\phit{general}$ guarantees that
every point represents exactly one tile and is in either $R_-$ or
$R_+$. From $M\models\phit{line}$ we can conclude that
\begin{iteMize}{$\bullet$}
\item 
every class of $\sim_2$ consists of exactly $2n$ points, one for each
possible combination of column number and being in $R_-$ or $R_+$;
\item points with adjacent column numbers have tiles that are
  compatible with respect to $H$; and
\item every $R_-$-point has the same tile as the $R_+$-point with the
  same column number.
\end{iteMize}
Thus, in particular, every class of $\sim_2$ represents a horizontally valid
row. 

From  $M\models\phit{start}$ it follows that there is a line that
represents a valid bottom row of a tiling. We finally show that, for
every line $l$ of $M$ it holds that $l$ represents a valid top row of
a tiling or there is another line $l'$ above $l$ that represents a
tiling row that is consistent with the tiling row of $l$ with respect
to $V$. 

It should be noted that we can not guarantee by an $\FO^2$-formula that
for every point $p$ in $R_+$ and column number 0 there is \textit{exactly} one
(full or partial) line in northwestern direction. However, the first
part of  $\phit{next-line}$ guarantees that there is \textit{at least} one line $l'$
such that all $R_-$-points of $l'$ are in northwestern direction of
$p$ (and therefore in northwestern direction of all $R_+$-points in
the line of $p$). Thus, the second part of $\phit{next-line}$
guarantees that the tiling row represented by $l'$ is compatible with
the row represented by $p$'s line. 

Therefore, $M$ has a line encoding a valid bottom row of a tiling and
every line either encodes a valid top row or has a line above that
encodes a compatible row with respect to $V$. As $M$ is finite, we can
conclude that there exists a valid tiling for $I$. The extraction
of a valid tiling from a model of $\varphi$ is illustrated in Figure
\ref{theo:hardness:fig}.

  \newcommand{\dotline}[4]
  {
    \pnode(#2,#1){f1}
    \pnode(#3,#1){f2}
    \ncline[linestyle=dotted]{*-*}{f1}{f2}
    \nbput{#4}
  }

  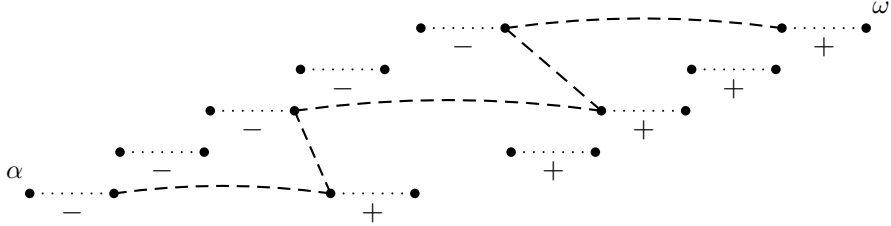
\begin{figure}[t] 
      \psset{xunit=0.4cm,yunit=0.55cm,labelsep=1mm}

\centering
    \begin{pspicture}(28,5)
      \rput(-0.3,0.5){\small $\alpha$}
      \rput(28.5,4.5){\small $\omega$}
      \dotline{0}{0.2}{3}{$-$}
      \dotline{1}{3.2}{6}{$-$}
      \dotline{2}{6.2}{9}{$-$}
      \dotline{3}{9.2}{12}{$-$}
      \dotline{4}{13.2}{16}{$-$}

      \dotline{0}{10.2}{13}{$+$}
      \dotline{1}{16.2}{19}{$+$}
      \dotline{2}{19.2}{22}{$+$}
      \dotline{3}{22.2}{25}{$+$}
      \dotline{4}{25.2}{28}{$+$}

      \pcarc[linestyle=dashed]{-}(3,0)(10.2, 0)
      \psline[linestyle=dashed]{-}(10.2, 0)(9,2)
      \pcarc[linestyle=dashed]{-}(9,2)(19.2,2)
      \psline[linestyle=dashed]{-}(19.2,2)(16, 4)
      \pcarc[linestyle=dashed]{-}(16,4)(25.2,4)
    \end{pspicture} 
    \caption{A sketch of the extraction of the rows of a valid tiling
      from a  model of formula $\varphi$. The dotted lines indicate
      rows represented by (half) a line. The dashed lines indicate
      the extracted tiling. \label{theo:hardness:fig}}
  \end{figure}
\end{myproof}


\subsection{Undecidable Extensions}

Next, we show that the approach of Theorem \ref{theo:twodcompl} fails
if the linear order is replaced by a (second) total preorder. It can
be concluded that Proposition \ref{prop:lpp} does not hold any more if
we allow vertical directions in $\exists$-constraints of $\LPP$. 

\begin{theorem}\label{theo:twopreorders}
  Finite satisfiability for $\FO^2(\precsim_1, \precsim_2)$ is undecidable.
\end{theorem}

\begin{myproof}
Let $\Sigma=\{0,1\}$.  We give a reduction from the Post
Correspondence Problem \PCP to finite satisfiability for
$\FO^2(\precsim_1,\precsim_2)$, where \PCP is defined as follows

  \problemdescr{\PCP}{A sequence $(u_1, v_1), \ldots, (u_k, v_k)$,
    where every $u_i,v_i \in \{0,1\}^*$.}{Is there a non-empty, finite
    sequence $\vec{i}=i_1, \ldots, i_m$ such that\newline 
$u_{i_1}    \ldots u_{i_m} = v_{i_1} \ldots v_{i_m}$?} 

  Let $I=(u_1, v_1), \ldots, (u_k, v_k)$ be an instance of the PCP. 

We first describe how solutions for $I$ can be encoded as
structures with two total preorders $\precsim_1, \precsim_2$ and some
unary relations. Then we describe how an $\FO^2(\precsim_1,
\precsim_2)$-sentence $\varphi_I$ can be constructed from $I$ that
describes properties that a model that encodes a solution to $I$
should fulfill. Finally, we show that from a model of $\varphi_I$ a
solution to $I$ can be extracted, thus establishing that
$I\mapsto\varphi_I$ is a reduction from \PCP to finite satisfiability
of $\FO^2(\precsim_1, \precsim_2)$-sentences.

For the first step, let $\vec{i}=i_1, \ldots, i_m$ be a solution to $I$ with
\mbox{$u := u_{i_1} \ldots u_{i_m} = v_{i_1} \ldots v_{i_m} =: v$}. The
universe of the intended model consists of all positions of $u$ and
all positions of $v$ and therefore has $|u|+|v|=2|u|$ elements. We
refer to the $i$-th position of $u$ and $v$ by $u[i]$ and $v[i]$,
respectively.\footnote{It should be stressed that $u[i]$ does not
  denote a symbol but a position in a string.}

The equivalence classes of $\precsim_1$ are simply the
sets $S_i=\{u[i],v[i]\}$, for $i\le|u|$, and, for all $i,j\le |u|$, if $i< j$ then $S_i\precsim_1 S_j$. 

The total preorder $\precsim_2$ has one equivalence class $C_j$ for every
$j\le m$ containing all positions corresponding to $u_{i_j}$ in $u$
and $v_{i_j}$ in $v$. Again, if $i < j$ then $C_i\precsim_2 C_j$ for all $i,j\le m$.

  The construction of $\preceq_1$ and $\preceq_2$ is illustrated in Figure \ref{theo:undec:fig}.

  \setlength{\dashlinegap}{3pt}
  \setlength{\dashlinedash}{0.5pt}
  \newdimen\mywidth
  \mywidth=1cm

  \begin{figure}[t]
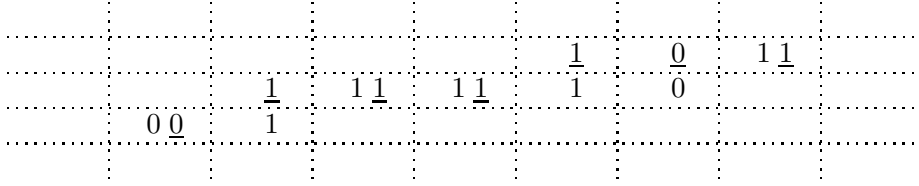
 \centering
      \begin{tabular}{p{\mywidth}: p{\mywidth}: p{\mywidth}: p{\mywidth}: p{\mywidth}: p{\mywidth}: p{\mywidth}: p{\mywidth}: p{\mywidth}}
	 $ $& &  &  &  &  &  & &$ $\\
	\hdashline
	 & &  &  &  & $\;\;\;\;\;\underline{1}$ & $\;\;\;\;\;\underline{0}$ & $\;\;\;1\;\underline{1}$& \\
	\hdashline
	 &  & $\;\;\;\;\;\underline{1}$ &$\;\;\;1\;\underline{1}$ &$\;\;\;1\;\underline{1}$ & $\;\;\;\;\;1$ & $\;\;\;\;\;0$ & & \\
	\hdashline
	 &$\;\;\;0\;\underline{0}$ & $\;\;\;\;\;1$ &  &  &  &  & & \\
	\hdashline
	 & &  &  &  &  &  & & \\
      \end{tabular}
    \caption{How the valid sequences $u := 01|1110|1$ and $v := 0|111|101$ are represented in a model for the $\FO^2(\precsim_1, \precsim_2)$-formula $\varphi$. Columns represent equivalence classes of $\precsim_1$ and rows represent equivalence classes of $\precsim_2$. Letters from $v$ are underlined. \label{theo:undec:fig}}
  \end{figure}

Furthermore, the intended model uses unary relation symbols from
$\{A_0,A_1\} \cup W \cup P \cup \{U, V\}$, where
\begin{iteMize}{$\bullet$}
\item $W=\{W_1,\ldots,W_k\}$, and
\item $P=\{P_1, \ldots, P_\ell\}$, where $\ell$ is the length of the
  longest word in $\{u_1, \ldots, u_k, v_1, \ldots, v_k\}$. 
\end{iteMize}
The relation $U$ consists of all positions of $u$ and $V$ consists of
all positions of $v$. Every position with symbol $0$ is in
$A_0$ and every position with symbol $1$ is in
$A_1$. Furthermore, if in $u = u_{i_1} \ldots u_{i_m}$, position $i$ of $u$ corresponds to
the $p$-th position of $u_{i_k}$ and $i_k=j$ then $u[i]$ is put into the
relations $W_j$ and $P_p$. Likewise, if in $v = v_{i_1} \ldots
v_{i_m}$, position $i$ of $v$ corresponds to the $p$-th position of
$v_{i_k}$ and $i_k=j$ then $v[i]$ is put into the relations $W_j$ and $P_p$.

This completes the description of the construction of the intended
model $M(\vec{i})$. 

Next, we give some necessary conditions that $M(\vec{i})$ has to
fulfill. 

   It is easy to construct an $\FO^2$-formula
   $\varphi_\text{symbols}$ expressing that
 \begin{iteMize}{$\bullet$}
    \item in every equivalence class of $\precsim_1$, there is exactly
      one element labeled from $U$ and one element from $V$, and
    \item the elements of every equivalence class of $\precsim_1$
      coincide on the predicates $A_0$ and $A_1$.
  \end{iteMize}
It should be noted here that in general it can not be expressed by
any $\FO^2$ formula that all classes of a given equivalence relation
have size two. However, in our setting it is sufficient to express
that no two different  elements from $U$ are $\sim_1$-equivalent and
likewise for $V$. 

  It is also not hard to construct an $\FO^2$-formula $\varphi_\text{words}$
  expressing that for every equivalence class $E$ of $\precsim_2$, there
  is some $j$ such that
  \begin{iteMize}{$\bullet$}
  \item all elements of $E$ are in $W_j$,
  \item for each $p\le |u_j|$ there is exactly one element $e$ in $E \cap P_p\cap U$ and $e$ is in $A_\sigma$ if and only if  $u_j[p]$ is labeled with $\sigma$, and  
  \item for each $p\le |v_j|$ there is exactly one element $e$ in $E \cap P_p\cap V$ and $e$ is in $A_\sigma$ if and only if  $u_j[p]$ is labeled with $\sigma$.
  \end{iteMize}

  Finally, an $\FO^2$-formula $\phit{order}$ can be
  constructed expressing that the two total preorders are consistent with each
  other, that is, 
  \begin{iteMize}{$\bullet$}
    \item for all elements $e_1, e_2 \in U$, if $e_1 \prec_2 e_2$, then
      $e_1 \prec_1 e_2$, and 
    \item for all elements $e_1, e_2 \in U$, if $e_1 \prec_1 e_2$ then
	\begin{iteMize}{$-$}
	  \item either $e_1 \prec_2 e_2$ 
	  \item or $e_1 \sim_2 e_2$ and $e_1\in P_k$, $e_2\in
            P_\ell$, for some $k<\ell$.
	\end{iteMize}
  \end{iteMize}
 Likewise for $v$-positions.

We let $\varphi_I =  \varphi_\text{symbols} \wedge
\varphi_\text{words} \wedge \varphi_\text{order}$. Clearly, if $\vec{i}$
is a solution for $I$, then $M(\vec{i})\models \varphi_I$. 

It remains to show that $I$ has a solution if there is a structure $M$
with $M\models\varphi_I$. To this end, let $M$ be a model for
$\varphi_I$. As $M\models \varphi_\text{words}$, we can associate with
every equivalence class $E$ of $\precsim_2$ an index $i(E)$ such that its
positions are in $W_{i(E)}$. Thus, the linear order $E_1\prec_2\cdots
\prec_2 E_m$ induces a sequence $\vec{i}=i_1,\ldots,i_m$, via
$i_j:=i(E_j)$. We show now that $\vec{i}$ is a solution to $I$. 

Let $a\not=b$ two elements from $U$. Formula
$\varphi_\text{symbols}$ ensures that they are not in the same
$\sim_1$-class. Thus, $\prec_1$ induces a linear order on the
$U$-positions and therefore these positions naturally constitute a string
$u$. Likewise,  the $V$-positions constitute a string
$v$. Formula $\varphi_\text{order}$
ensures that $a\prec_1 b$ if $a\prec_2 b$ or if $a\sim_2 b$ and there
are $k<l$ with $a\in P_k$ and $b\in P_l$. Thus,
$u=u_{i_1}\cdots u_{i_m}$ and, likewise, $v=v_{i_1}\cdots
v_{i_m}$. Finally, $ \varphi_\text{symbols}$ guarantees that $|u|=|v|$
and, for every $i$, $u[i]$ and $v[i]$ carry the same symbol. This completes the proof of the
correctness of the reduction. Clearly, $I\mapsto \varphi_I$ can be
computed (even in polynomial time).
\end{myproof}


Next we prove undecidability in case one \tpo is replaced by two
linear orders in the previous theorem.

\begin{theorem}\label{theo:twostrictonelinearorder}
  Finite satisfiability for $\FO^2(\leq_1, \leq_2, \precsim_3)$ is undecidable.
\end{theorem}

\begin{myproof}
  As in the previous proof we reduce from \PCP. However, given a
  \PCP-instance $I$ we first translate it into a modified instance
  $I'$ over the alphabet $\hat{\Sigma}=\{0,1,0',1'\}$. The idea is to replace
  every symbol $\sigma$ in $I$ by $\sigma\sigma'$ in $I'$, e.g., to
  transform the instance $(0,01),(10,1)$ into
  $(00',00'11'),(11'00',00')$. More formally, $I'=h(I)$ where $h$ is the
  homomorphism defined by $h(0)=00'$ and $h(1)=11'$.

We call letters of the form $\sigma'$
  \textit{marked}. Clearly, $I'$ has a solution if and
  only if $I$ has a solution.   

From $I'$ we will construct a   $\FO^2(\leq_1, \leq_2,
\precsim_3)$-formula $\varphi_I$ with the intention that $\varphi_I$
has a finite model if and only if $I$ has a solution. The construction
of $\varphi_I$ is similar as in the proof of Theorem
\ref{theo:twopreorders}. However, the role of $\precsim_1$ in that
proof will be mimicked by $\leq_1$ and $\leq_2$, here. 

  As before, we first describe how solutions for $I'$ can be encoded
  as structures with two linear orders $\leq_1, \leq_2$, a
  total preorder $\precsim_3$, and some unary relations. Then we
  define $\varphi_I$ and show that $I\mapsto\varphi_I$ is a reduction
  from \PCP to finite satisfiability of  $\FO^2(\leq_1, \leq_2, \precsim_3)$.

For the first step, let $\vec{i}=i_1, \ldots, i_m$ be a solution to $I$ with
$u := u_{i_1} \ldots u_{i_m} = v_{i_1} \ldots v_{i_m} =: v$. As before
the intended model $M(\vec{i})$ consists of all positions of $u$ and
all positions of $v$ and therefore has $2|u|$ elements.

The total preorder $\precsim_3$  is defined exactly as $\precsim_2$ in
the previous proof. The model has additional unary
relations $\{A_0,A_1,M\} \cup W \cup P \cup \{U, V\}$ where $U,V$ and
the relations from $W\cup P$ are defined as before. Positions with
symbol $\sigma\in\{0,1\}$ are in $A_\sigma$, positions with symbol
$\sigma'\in\{0',1'\}$ are in $A_\sigma$ and in $M$. 

The linear orders $<_1$ and $<_2$ are defined by 
$$u[1] <_1 v[1] <_1 u[2] <_1 v[2] <_1 \ldots <_1 u[|u|] <_1 v[|u|]$$
and 
$$v[1] <_2 u[1] <_2 v[2] <_2 u[2] <_2 \ldots <_2 v[|u|] <_2 u[|u|]$$
 
See \ref{theo:undec1:fig} for an illustration of the construction of $\le_1$ and $\le_2$.

  \setlength{\dashlinegap}{3pt}
  \setlength{\dashlinedash}{0.5pt}
  \newdimen\mywidth
  \mywidth=0.2cm

  \begin{figure}[t]
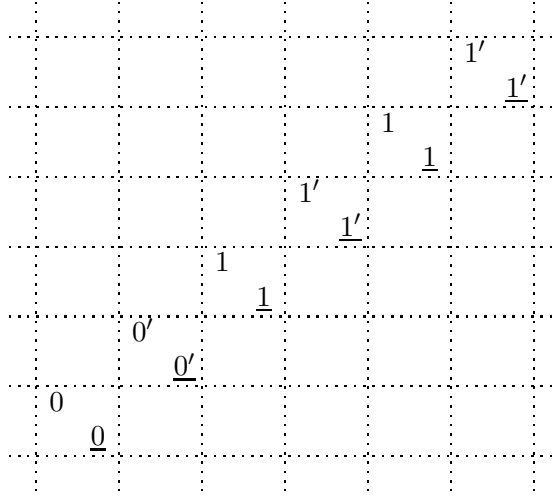
 \centering
      \begin{tabular}{p{\mywidth}: p{\mywidth} p{\mywidth}: p{\mywidth} p{\mywidth}: p{\mywidth} p{\mywidth}: p{\mywidth} p{\mywidth}: p{\mywidth} p{\mywidth}: p{\mywidth} p{\mywidth}: p{\mywidth}}
	& & & & & & & & & & & & & \\
	\hdashline
	& & & & & & & & & & & $1'$ & & \\
	& & & & & & & & & & & &  $\underline{1'}$& \\
	\hdashline
	& & & & & & & & & $1$ & & & & \\
	& & & & & & & & & & $\underline{1}$ & & & \\
	\hdashline
	& & & & & & & $1'$ & & & & & & \\
	& & & & & & & & $\underline{1'}$& & & & & \\
	\hdashline
	& & & & & $1$ & & & & & & & & \\
	& & & & & & $\underline{1}$ & & & & & & & \\
	\hdashline
	& & & $0'$ & & & & & & & & & & \\
	& & & & $\underline{0'}$ & & & & & & & & & \\
	\hdashline
	& $0$ & & & & & & & & & & & & \\
	& & $\underline{0}$ & & & & & & & & & & & \\
	\hdashline
	& & & & & & & & & & & & & \\

      \end{tabular}
    \caption{How the intended linear orders $\leq_1, \leq_2$ for a valid sequences $u := 00'|11'11'$ and $v := 00'11'|11'$ look like. Columns are ordered by $<_1$ and rows are ordered by $<_2$. Letters from $v$ are underlined. \label{theo:undec1:fig}}
  \end{figure}

  Now we describe necessary conditions that $M(\vec i)$ has to fulfill.

First of all, it should satisfy formula $\phit{words}$ of
the previous proof, referring to $\precsim_3$ in place of  $\precsim_2$.

Let $\varphi_{3}(x,y)$ be the formula which expresses that
\begin{iteMize}{$\bullet$}
\item either $x\prec_3 y$
\item or $x\sim_3 y$ and $x\in P_i$ and $y\in P_j$, for some $i<j$.
\end{iteMize}
It is not hard to see that in every model of $\phit{words}$
the formula $\varphi_3$ induces a linear order on the positions in $U$
and on the positions in $V$. Furthermore,  $\phit{words}$
guarantees that both orders begin with unmarked positions, alternate between marked and
unmarked positions and end with a marked position (just because the
strings $u_i, v_i$ do so). 
In the following, we write $x<_3 y$ as an abbreviation for $\varphi_{3}(x,y)$.

Using $\varphi_{3}$, we can define a formula
$\phit{2-order}$ expressing that $\le_1$, $\le_2$ and $<_3$
 are consistent on positions in $U$ and on positions in $V$. That
is, if $x,y\in U$ then $x<_1 y$ if and only if $x<_2 y$ if and only if
$x<_3 y$ and likewise for $x,y\in V$. In the following, we denote
these two orders on the positions of $U$ and $V$ simply by $\le$.

Finally, we construct a formula $\phit{symbols}$ with a
  similar intention as in the proof  of Theorem \ref{theo:twostrictonelinearorder}.

To this end, let $\phit{bi}(x,y)$ be a formula expressing that
$x\le_1 y$ and $y\le_2 x$.

The formula $\phit{symbols}$ expresses that the following
conditions hold.
\begin{iteMize}{$\bullet$}
\item $\phit{bi}(x,y)$ only holds for positions $x\not=y$ if
  \begin{iteMize}{$-$}
  \item $x\in U$ and $y\in V$, and
  \item both carry the same symbol from $\hat{\Sigma}$.
  \end{iteMize}
\item For every position $x$, there is a position $y$ such that
  $\phit{bi}(x,y)$ holds and vice versa.
\end{iteMize}

We let $\varphi_I =  \varphi_\text{symbols} \wedge
\varphi_\text{words} \wedge \varphi_\text{2-order}$ and it is again easy
to see that, if $\vec{i}$ is a solution for $I$, then $M(\vec{i})\models \varphi_I$.  

It remains to show that $I$ has a solution if there is a structure $M$
with $M\models\varphi_I$. To this end, let $M$ be a model for
$\varphi_I$. As $M\models \varphi_\text{words}$, we can associate with
every equivalence class $E$ of $\precsim_3$ an index $i(E)$ such that its
positions are in $W_{i(E)}$, just as in the previous proof. 
Again, the linear order $E_1\prec_3\cdots
\prec_3 E_m$ induces a sequence $\vec{i}=i_1,\ldots,i_m$, via
$i_j:=i(E_j)$. We show now that $\vec{i}$ is a solution to $I$.

For this purpose, it is sufficient to show that
\begin{enumerate}[(1)]
\item $\phit{bi}(x,y)$ defines
a bijection between the positions in $U$ and the positions in $V$, and
\item that this bijection
is compatible with $\le$, that is, if
$\phit{bi}(a_1,b_1)$ and $\phit{bi}(a_2,b_2)$ hold for
$a_1\not=a_2\in U$ and $b_1\not=b_2\in V$ then $a_1< a_2$ if and only if
$b_1<b_2$.  
\end{enumerate}

Indeed, from (1) and (2) it follows that the bijection induced by
$\phit{bi}$ pairs the $i$-th position of $U$ with the $i$-th
position of $V$, for every $i$. As $\phit{symbols}$
guarantees that corresponding positions carry the same symbol, it
follows that $\vec{i}$ is a solution. 

We next show (1), that is, there are no positions $a_1\not=a_2\in U$ and $b\in V$
such that $\phit{bi}(a_1,b)$ and $\phit{bi}(a_2,b)$
hold (and, correspondingly, for $a,b_1,b_2$). Towards a contradiction let
us assume the existence of such elements $a_1\not=a_2\in U$ and $b\in V$
for which $\phit{bi}(a_1,b)$ and $\phit{bi}(a_2,b)$
hold. Let, without loss of generality $a_1< a_2$ and let us assume
that both are marked.\footnote{We recall that
  $\phit{symbols}$ guarantees that $a_1$, $a_2$ and $b$
  are either all marked or all unmarked.} Then there is an unmarked
$U$-position $a$ with $a_1<a<a_2$. But then we can conclude that
$a\le_1 b$ and $b \le_2 a$, a contradiction as
$\phit{symbols}$ holds in $M$ and $a$ is unmarked and $b$ is
marked. 

To show (2) let us assume, again towards a contradiction, that there
are $a_1\not=a_2\in U$ and $b_1\not=b_2\in V$ such that
$\phit{bi}(a_1,b_1)$ and $\phit{bi}(a_2,b_2)$,
$a_1<a_2$ and $b_2<b_1$. But then
\begin{iteMize}{$\bullet$}
\item $a_1<_1 a_2 \le_1 b_2$ and
\item $b_2<_2 b_1 \le_2 a_1$ 
\end{iteMize}
and therefore $\phit{bi}(a_1,b_2)$ holds, contradicting
(1). Similarly, from $a_1>a_2$ and $b_2>b_1$ it follows that
$\phit{bi}(a_2,b_1)$ holds, again in contradiction to (1). 
\end{myproof}


\section{Conclusion}  \label{sec:conclusion}

The context of our results was already discussed in the introduction. Table \ref{tab:dwresults} summarizes the results of this paper
  and previous results for data
  words with different kinds of orders on positions and data values.

We 
mention some possible lines for extensions and further  research. We recall that $\FO^2(\leq_1,
  \precsim_2)$ is an abbreviation for $\FO^2(\FinOrd(\leq_1,
  \precsim_2))$.

  \subsubsection*{Two-Variable Logic.} The lower bound for finite
  satisfiability of $\FO^2(\leq_1; \precsim_2)$ from Theorem
  \ref{theo:expspacetwoorders} does not immediately carry over to
  $\FO^2(\Sigma,\leq_1, [\precsim_2,S_2])$ as the translation of a
  $\FO^2(\leq; \precsim)$ formula $\varphi$ into 
a $\DW^2(\Sigma,\leq, \precsim)$ formula might yield an alphabet of
exponential size in $|\varphi|$. Thus, there remains a gap between the
\EXPSPACE upper bound from Corollary \ref{cor:dw} and the $\NEXPTIME$ lower bound from \cite{BojanczykMuscholl+Two-Variable06}.
Further there is still a gap between the ``eight orders'' undecidability result of
    \cite{OttoTwo01} and the decidability for $\FO^2$ with two
    linear orders in this paper.

In the context of automated verification it would be interesting to
    generalize our results from data words to data $\omega$-words.

    \begin{table}[t]
{\small
      \centering
      \begin{tabular}{l*{3}{|p{3cm}}}
 Data $\backslash$ Positions& \centering Successor & \centering Linear
 Order & {\centering Successor \& \newline  Linear Order}\\
\hline
Successor
& \begin{center}\vspace{-3mm}{\centering ?}\end{center}
& \centering  in \EXPSPACE (this work) 
& undecidable \newline \cite{ManuelZ11}
 \\
Linear Order 
& \begin{center}\vspace{-3mm}{\centering ?}\end{center}
& \centering in \EXPSPACE (this work) 
& {\centering undecidable \cite{BojanczykMuscholl+Two-Variable06} }
\\
Successor \& LO
& \begin{center}\vspace{-3mm}{\centering ?}\end{center}
& \centering in \EXPSPACE (this work) 
& {\centering undecidable \cite{BojanczykMuscholl+Two-Variable06} }
\\
single-occ.\ Succ
&   \centering in 2\NEXPTIME \cite{Manuel10} 
& \centering in \EXPSPACE (this work) 
& {\centering decidable \newline \cite{ManuelZ11}}
\\
single-occ.\ LO 
& \centering in \EXPSPACE (this work) 
& \centering in \EXPSPACE (this work) 
& {\centering in \EXPSPACE  \newline (this work)} \\
     \end{tabular}
}
      \caption{Summary of results on Finite Satisfiability of $\FO^2$
        on data words with ordered data values. ``single-occ.'' refers
      to the case where each data values occurs at most once in a data
    word.}
        \label{tab:dwresults}
      
    \end{table}
  
 \subsubsection*{Other Logics.} There are connections between the
 results of this paper and some temporal logics, Compass Logic and
 Interval Temporal Logic. Some of these connections have been made
 precise in the conference paper underlying this article \cite{SchwentickZ10}.

  \textit{Compass Logic} is a two-dimensional temporal logic, whose
  operators allow for moving north, south, east and west along a grid
  \cite{Venema1990}. Satisfiability for compass logic is known to be
  undecidable \cite{MarxR1997}. Compass logic can be extended in two
  directions.  Up to now, only complete grids have been considered as
  underlying structures. Partial grids, i.e. grids where not all
  crossings need to exist, can be considered as underlying structures
  as well. Furthermore, operators northeast, northwest, southeast and
  southwest can be considered. The results of this article  can be
  used, after some appropriate modifications, to yield decidability for compass logic with these extensions when only the operators northwest, northeast, southwest, southeast, west and east are used.

  \textit{Interval Temporal Logic} can reason about intervals of time
  with the help of operators as `after', `during', `begins'
  etc. Expressions such as `Immediately after we finished writing the
  paper, we will go to the beach' can be captured. To this end,
  propositions, like ``writing the paper'' or ``go to the beach'', are
  assigned to time intervals. In conventional interval temporal logic,
  all possible intervals are considered as part of a structure, that
  is, reasoning is always with respect to all intervals. In a setting
  where structures may consist of a subset of the set of all
  intervals, decidability of reasoning with the operators `ends',
  `later' and `during' as well as their duals can be obtained from our
  results, again after appropriate modifications. 

  We believe that partial models for compass and interval logic deserve some further investigations.

Besides relations to compass logic and interval logic we conjecture that there are connections to spatial
reasoning that is done in the context of Geographical Information
Systems (for a survey, see \cite{CohnHazarikaQualitative01}). 


\bibliographystyle{amsalpha}
\bibliography{bibliography}

\end{document}